\documentclass[10pt,conference]{IEEEtran}
\IEEEoverridecommandlockouts
\usepackage[utf8]{inputenc}
\usepackage{amsmath,amssymb,amsfonts}
\usepackage{algorithmic}
\usepackage{graphicx}
\usepackage{textcomp}
\usepackage{xcolor}
\def\BibTeX{{\rm B\kern-.05em{\sc i\kern-.025em b}\kern-.08em
    T\kern-.1667em\lower.7ex\hbox{E}\kern-.125emX}}

\usepackage{booktabs,multirow}
\usepackage{subcaption} 

\usepackage{adjustbox}
\usepackage{tikz}
\usetikzlibrary{quantikz}
\usetikzlibrary{external}
\usepackage{amsthm}
\newtheorem{lemma}{Lemma}
\newtheorem{theorem}{Theorem}
\newtheorem{definition}{Definition}

\usepackage[colorlinks=true]{hyperref}

\usepackage[switch]{lineno}
\newcommand*\patchAMSlineno[1]{
	\expandafter\let\csname old#1\expandafter\endcsname\csname #1\endcsname
	\expandafter\let\csname oldend#1\expandafter\endcsname\csname end#1\endcsname
	\renewenvironment{#1}
		{\linenomath\csname old#1\endcsname}
		{\csname oldend#1\endcsname\endlinenomath}
}
\AtBeginDocument{
	\patchAMSlineno{equation}
	\patchAMSlineno{equation*}
	\patchAMSlineno{align}
	\patchAMSlineno{align*}
}

\newcommand{\dicke}[2]{\ket{\smash{D_{#2}^{#1}}}}
\newcommand{\SCS}{\mathit{SCS}}
\newcommand{\dsu}[2]{\mathit{U}_{#2}^{#1}}
\newcommand{\dsugate}[2]{\gate[#1]{\dsu{#1}{#2}}}

\newcommand{\wdb}[3]{\mathit{WDB}_{#3}^{#1,#2}}
\newcommand{\wdbgate}[3]{\gate[#1]{\wdb{#1}{#2}{#3}}}

\newcommand{\ryangle}[2]{\underset{\scriptstyle\smash{\surd#1/#2}}{\theta}}
\newcommand{\rygate}[2]{\gate{\ryangle{#1}{#2}}}

\newcommand{\bigO}{\mathcal{O}}
\DeclareMathOperator{\plog}{poly\,log}
\newcommand{\Id}{\textit{Id}}
\newcommand{\CNOT}{\textit{CNOT}}

\begin{document}

\title{Short-Depth Circuits for Dicke State Preparation%
	\thanks{Research presented in this article was supported by the Laboratory Directed Research and Development program of Los Alamos National Laboratory under project number 20220656ER.\hfill LA-UR-22-25501}
}
\author{
	\IEEEauthorblockN{Andreas Bärtschi}
	\IEEEauthorblockA{\textit{CCS-3 Information Sciences} \\
	\textit{Los Alamos National Laboratory}\\
	Los Alamos, NM 87544, USA \\
	baertschi@lanl.gov}
\and
    \IEEEauthorblockN{Stephan Eidenbenz}
	\IEEEauthorblockA{\textit{CCS-3 Information Sciences} \\
	\textit{Los Alamos National Laboratory}\\
	Los Alamos, NM 87544, USA \\
	eidenben@lanl.gov}
}

\maketitle

%\linenumbers

\begin{abstract}
	We present short-depth quantum circuits to deterministically prepare any Dicke state $\dicke{n}{k}$, which is the equal-amplitude superposition of all $n$-qubit computational basis states with Hamming Weight $k$.
	Dicke states are an important class of entangled quantum states with a large variety of applications, and a long history of experimental creation in physical systems.
	On the other hand, not much is known regarding efficient scalable quantum circuits for Dicke state preparation on realistic quantum computing hardware connectivities. 

	Concretely, we present preparation circuits for Dicke states $\dicke{n}{k}$ with 
	(i) depth of $\bigO(k \log \tfrac{n}{k})$ for all-to-all connectivity (such as on current ion trap devices);
	(ii) depth of $\bigO(k \surd \tfrac{n}{k})=\bigO(\sqrt{nk})$ for Grid connectivity on grids of size $\Omega(\surd \tfrac{n}{s}) \times \bigO( \sqrt{ns})$ with $s \leq k$ (such as on most current superconducting qubit devices). 
	
	Both approaches have a total gate count of $\bigO(kn)$, need no ancilla qubits, 
	and generalize to preparation and compression of those symmetric pure states, in which all non-zero amplitudes correspond to states with Hamming weight at most $k$.
	Our work significantly improves and expands previous state-of-the art circuits which had depth $\bigO(n)$ on a Linear Nearest Neighbor connectivity for arbitrary $k$ (FCT 2019~\cite{baertschi2019deterministic}) % (Fundamentals of Computation Theory 2019) 
	and depth $\bigO(\log n)$ on all-to-all connectivity for $k=1$ (AQT 2019~\cite{epfl2019wstate}). %(Advanced Quantum Technologies 2019).
\end{abstract}


\section{Introduction}
\label{sec:intro}

Dicke states are among the most useful highly-entangled quantum states. A Dicke state \dicke{n}{k} is defined as the equal-amplitude superposition of all $n$-qubit computational basis states $x$ of Hamming weight $\mathrm{HW}(x)=k$, 
where the Hamming weight $\mathrm{HW}(x)$ denotes the number of 1s in the bitstring~$x$, 
\begin{align*}
	\label{eq:dicke}
	\dicke{n}{k} = \tbinom{n}{k}^{-\frac{1}{2}} \sum\nolimits_{x \in \left\{ 0,1 \right\}^n,\ \mathrm{HW}(x)=k}{\ket{x}}.
\end{align*}
For example, $\dicke{4}{2} = \tfrac{1}{\surd{6}}(\ket{1100}+\ket{1010}+\ket{1001}+\ket{0110}+\ket{0101}+\ket{0011}).$
We note that by flipping every bit, we have $\dicke{n}{n-k} = X^{\otimes n}\dicke{n}{k}$, thus results for $\dicke{n}{k}$ also hold for $\dicke{n}{n-k}$
and we may assume $2k \leq n$ without loss of generality.
Dicke states~\cite{Dicke1954} have been studied in quantum game theory~\cite{Oezdemir2007}, quantum networking~\cite{Prevedel2009}, quantum metrology~\cite{Toth2012,ouyang2021robust}, 
quantum error correction~\cite{ouyang2014permutation,ouyang2021permutation} and quantum storage~\cite{ouyang2021quantum}. 
From a quantum algorithms perspective, their high-impact application is their use as initial superpositions of all feasible solution states of Hamming-weight constrained 
combinatorial optimization problems such as Maximum k-Vertex Cover. Indeed such use cases have been explored for adiabatic computing~\cite{Childs2002} and for variational quantum optimization, 
in particular the Quantum Alternating Operator Ansatz (QAOA) \cite{Hadfield2019,nasa2020XY,cook2020kVC,baertschi2020grover,golden2021grover,golden2022evidence}.
Due to their wide applicability, various smaller Dicke states have been implemented on most NISQ technologies such as trapped ions~\cite{Hume2009,Ivanov2013,Lamata2013}, atoms~\cite{Stockton2004,Xiao2007,Shao2010},
photons~\cite{Prevedel2009,Wieczorek2009}, superconducting qubits~\cite{Wu2016,aktar2022divideconquer}, and others~\cite{johnsson2020geometric,wu2019initializing}. Finally, Dicke states are a useful entanglement 
benchmark test (as alternative or in addition to GHZ states) for novel quantum devices~\cite{somma2006fidelity}.

\subsection*{Results and Relevance} In this paper, we propose scalable short-depth quantum circuits to deterministically prepare Dicke states $\dicke{n}{k}$ on gate-based quantum computers. 
More precisely, we first present preparation circuits for Dicke states $\dicke{n}{k}$ with a depth of $\bigO(k \log \tfrac{n}{k})$ if the underlying qubit topology is an all-to-all connectivity. 
Such connectivity is characteristic of current ion trap devices, such as those of Honeywell/Quantinuum or IonQ. 

A second set of circuits achieves a depth of $\bigO(k \surd \tfrac{n}{k}) = \bigO(\sqrt{nk})$ for grid connectivity, where we prefer the former notation to illustrate the constructive similarity to the first result.
Our grid connectivity result holds for any ``grid-like'' topology of size $\Omega(\surd \tfrac{n}{s}) \times \bigO( \sqrt{ns})$ with $1\leq s \leq k$, which is typical for most current superconducting qubit devices, such as those of IBMQ, Google and Rigetti.
We call a topology (e.g., a heavy-hex lattice) grid-like, if there is a canonical mapping from a grid to the topology with at most a constant-factor overhead 
in both depth and 2-qubit gate count when adapting circuits along. 
These topologies are of long-term interest as error-correction codes such as the honeycomb code~\cite{Hastings2021dynamically,Haah2022boundarieshoneycomb},
the surface code~\cite{dennis2002topological,fowler2012surface} and hybrids~\cite{Chamberland2020codes} with Bacon-Shor~\cite{bacon2006code,shor1995code} assume such connectivities. 
For simplicity, in our proofs we assume Cartesian grid topologies.

Both of our approaches have a total gate count of $\bigO(kn)$, need no ancilla qubits, 
and generalize to preparation and compression of those symmetric pure states, in which all non-zero amplitudes correspond to states with Hamming weight at most $k$.
The key ingredient in our circuit design is the classical algorithm design paradigm of recursive divide-and-conquer.
Based on (perhaps misleading) insights from classical lower bounds, we conjecture our circuit depths are optimal up to constant factors and low-order terms 
for constant $k$ on all-to-all, and for arbitrary $k$ on $\surd{\tfrac{n}{k}} \times \sqrt{nk}$ grid topologies.

\begin{table*}[t] 
    \centering
	\begin{tabular}{@{}llllllr@{}}
		\toprule State Preparation
		& Reference							& State Type			& Circuit Depth				& \# of {\CNOT} Gates			& \# of Ancilla			& Topology		\\
		\toprule Probabilistic                                                                                                                                  	                                       
		& Childs \textit{et.al.} 2000~\cite{Childs2002}			& Dicke states			& $\bigO(n\ \plog n)$			& $\bigO(n\ \plog n)$			& $\bigO(\log n)$		& LNN			\\	
		\midrule Deterministic                                                                                                                                  	                                       
		& Cruz \textit{et.al.} 2018~\cite{epfl2019wstate}		& W states			& $\bigO(\log n)$ 			& $\bigO(n)$				& --				& all-to-all		\\
		& Kaye, Mosca 2004~\cite{Mosca2001}				& Symmetric states		& $\bigO(n\ \plog(n/\varepsilon))$	& $\bigO(n\ \plog(n/\varepsilon))$	& $\bigO(\log(n/\varepsilon))$	& LNN			\\			
		& Bärtschi, Eidenbenz 2019~\cite{baertschi2019deterministic}	& Dicke states			& $\bigO(n)$				& $\bigO(kn)$				& --				& LNN			\\
		& 								& Symmetric states		& $\bigO(n)$				& $\bigO(n^2)$				& --				& LNN			\\
		                                                                                                                                        
		& \textbf{This work}						& \textbf{Dicke states}		& \boldmath$\bigO(k \log \tfrac{n}{k})$	& \boldmath$\bigO(kn)$			& --				& \textbf{all-to-all}	\\
		&								& \textbf{Dicke states}		& \boldmath$\bigO(k \surd\tfrac{n}{k})$	& \boldmath$\bigO(kn)$			& --				& \textbf{Grid}		\\
		\midrule (Un)Compression                                                                                                                                  	                                       
		& Bacon, Harrow, Chuang 2004~\cite{Bacon2006}			& Schur Transform		& $\bigO(n\ \plog(n/\varepsilon))$	& $\bigO(n\ \plog(n/\varepsilon))$	& $\bigO(\log(n/\varepsilon))$	& LNN			\\
		& Plesch, Bu\v{z}ek 2009~\cite{Plesch2010}			& Symmetric states		& $\bigO(n^2)$				& $\bigO(n^2)$				& --				& LNN			\\	
		\midrule Sparse States
		& Zhang, Li, Yuan 2022~\cite{zhang2022quantum}			& $d$-sparse			& $\bigO(\log(dn))$ 			& $\bigO(dn \log (dn))$			& $\bigO(dn \log d)$		& all-to-all		\\
		\bottomrule
	\end{tabular}
	\caption{%
		State preparation schemes for Dicke states as well as subsets and supersets thereof: W states are Dicke states of Hamming weight 1, while Symmetric states are superpositions of Dicke states. 
		Probabilistic state preparation uses a projective measurement of a $n$-qubit product state into a Hamming weight subspace.
		Quantum compression is more general and is used in reverse for state preparation. Here $\varepsilon$ denotes the error coming from the precision of reversible floating-point arithmetic circuits.
		Sparse state preparation is also more general, but applied to Dicke states which are $d=\binom{n}{k}$-sparse, we have $\log(d) \geq k\log\tfrac{n}{k}$, 
		giving the same circuit depth as our first result, but with a superpolynomial number of gates and ancilla.
	}
	\label{tbl:asymptotics}
\end{table*}

\subsection*{Paper Organization}  Our Dicke state preparation circuits of depth $\bigO(k \log \tfrac{n}{k})$ and $\bigO(k \surd\tfrac{n}{k})$ for all-to-all and grid topologies, respectively, 
significantly improve upon and extend previous state-of-the art circuits that had 
depth $\bigO(n)$ on a Linear Nearest Neighbor connectivity for arbitrary~$k$~\cite{baertschi2019deterministic} and 
depth $\bigO(\log n)$ on all-to-all connectivity for $k=1$ \cite{epfl2019wstate}.

We give a detailed overview over other relevant results in Section~\ref{sec:related}. 
In Section~\ref{sec:preliminaries}, we establish necessary preliminaries and present the main recursive divide-and-conquer idea, including the ``conquer'' part of the scheme.
The ``divide'' part, which requires only LNN connectivity -- and thus a subgraph of all-to-all and grid topologies -- is presented in Section~\ref{sec:divide}.
The topology-dependent recursive structure is presented for both all-to-all and grid topologies in Section~\ref{sec:recursive}.
We finish with a conclusion, Section~\ref{sec:conclusion}.

\section{Related Work}
\label{sec:related}

Table~\ref{tbl:asymptotics} gives an overview over relevant results and a comparison to our new results. We need to introduce related states: $W$ states are Dicke states of Hamming weight 1, $\dicke{n}{1}$,
and Symmetric states are superpositions of Dicke states. They can thus be seen as a subset (superset, respectively) of Dicke states. Furthermore, Dicke states have $d=\binom{n}{k}$ non-zero state vector entries. 
Thus we can consider them to be $d$-sparse, though for non-constant $k$ the sparsity scales superpolynomially.
There are four basic approaches to state preparation of these states: probabilistic, deterministic, through quantum (un-)compression via the Schur Transform, and through sparse state preparation techniques.

Probabilistic state preparation uses a projective measurement of an $n$-qubit product state into a Hamming-weight subspace as proposed by \cite{Childs2002} which succeeds with probability 
$\tbinom{n}{k}(\tfrac{k}{n})^k (\tfrac{n-k}{n})^{n-k} \in \Omega(\tfrac{1}{\surd n})$, resulting in a $\bigO(n\ \plog n)$ circuit depth 
with $\bigO(n\ \plog n)$ CNOT gates and requiring the use of $\bigO(\log n)$ ancilla qubits on a Linear Nearest Neighbor (LNN) qubit topology. 

Dicke states can also be generated by using the (inverse) Schur Transform, which implements the Schur-Weyl duality between the computational basis and the Schur basis 
with a circuit depth of  $\bigO(n\ \plog(n/\varepsilon))$ and $\bigO(\log(n/\varepsilon))$ ancilla qubits on a LNN topology, where $\varepsilon$ denotes the error coming from the precision of floating-point arithmetic subcircuits~\cite{Bacon2006}. 
For qubits, the Schur basis gives the total angular momentum $J$ and its $z$-component $m$ of a state, plus an additional index to distinguish between states of the same $(J,m)$ tuple.
In case of Dicke states~$\dicke{n}{k}$, we have unique values $J=n$ and $m=n-2k$~\cite{somma2006fidelity}. Thus the Schur transform can be used for quantum compression of Symmetric states,
for which a construction was given with quadratic depth and notably without ancilla qubits~\cite{Plesch2010}, by removing the need to store $J$ and index values.

In deterministic state preparation, where a circuit deterministically prepares a Symmetric state, an initial proposal~\cite{Mosca2001} builds circuits of depth $\bigO(n\ \plog(n/\varepsilon))$ 
using $\bigO(\log(n/\varepsilon))$ ancilla qubits on an LNN topology, matching the result later obtained with the Schur transform. 
An improvement~\cite{baertschi2019deterministic} leads to state preparation circuits of linear depth without the need for ancilla qubits,
requiring $\bigO(n^2)$ CNOT gates for Symmetric states and $\bigO(kn)$ CNOT gates for Dicke states on LNN connectivities.
For $W$ states, a logarithmic $\bigO(\log n)$ construction exists for all-to-all connectivities~\cite{epfl2019wstate}. This scheme can easily be extended to $\bigO(\surd n)$ depth for grid connectivities.
Several approaches aimed at NISQ implementation efficiency (for small $n$ and $k$) feature improvements that are not reflected in asymptotics \cite{mukherjee2020preparing,aktar2022divideconquer}. 

For arbitrary $d$-sparse states, state preparation schemes were recently presented~\cite{zhang2022quantum} requiring only the optimum depth of $\bigO(\log(dn))$, 
albeit at a high number of $\bigO(dn \log d)$ ancilla qubits. Applied to Dicke states which are $d=\binom{n}{k}$-sparse, 
we have depth $\log(dn) \geq \log(d) \geq k\log\tfrac{n}{k}$ using in total a superpolynomial number of gates and ancilla.

Our results are the first Dicke state preparation circuits that are sublinear in depth with respect to $n$,
without the need for superpolynomial gate or ancilla resources, see Table~\ref{tbl:asymptotics}.

\begin{figure*}[t!]
	\centering
	\begin{adjustbox}{width=\linewidth}
	%% Grover Mixer QAOA
	%% -----------------
	%\hspace*{-0.5cm}
		\begin{quantikz}[row sep={24pt,between origins},execute at end picture={}]
			\lstick{\ket{0}}	& \rygate{1}{20}	& \ctrl{1}		& \qw\slice{}		& \qw		& \qw		& \octrl{5}\slice{}	& \dsugate{3}{3}	& \qw\rstick[6]{\rotatebox{90}{$\dicke{6}{3}$}} & 	& 	& 	& 	& 	& \lstick{$\ket{0^{1-\ell}1^{\ell}}$}	& \ctrl{4}	& \targ{}\slice{}	& \ctrl{2}	& \targ{}	& \ctrl{1}	& \targ{}\slice{}	& \dsugate{1}{1}	& \qw\rstick[6]{\rotatebox{90}{$\dicke{6}{\ell}$}} 	\\ 
			\lstick{\ket{0}}	& \qw			& \rygate{9}{19}	& \ctrl{1}		& \qw		& \octrl{3}	& \qw			& \qw			& \qw					 	& 	& 	& 	& 	& 	& \lstick{\ket{0}}			& \qw		& \qw			& \qw		& \qw		& \rygate{1}{2}	& \ctrl{-1}		& \dsugate{1}{1}	& \qw					 		\\
			\lstick{\ket{0}}	& \qw			& \qw			& \rygate{9}{10}	& \octrl{1}	& \qw		& \qw			& \qw			& \qw 						& 	& 	& 	& 	& 	& \lstick{\ket{0}}			& \qw		& \qw			& \rygate{2}{4}	& \ctrl{-2}	& \ctrl{1}	& \targ{}		& \dsugate{1}{1}	& \qw 							\\
			\lstick{\ket{0}}	& \qw			& \qw			& \qw			& \targ{}	& \qw		& \qw			& \dsugate{3}{3}	& \qw 						& 	& 	& 	& 	& 	& \lstick{\ket{0}}			& \qw		& \qw			& \qw		& \qw		& \rygate{1}{2}	& \ctrl{-1}		& \dsugate{1}{1}	& \qw 							\\
			\lstick{\ket{0}}	& \qw			& \qw			& \qw			& \qw		& \targ{}	& \qw			& \qw			& \qw 						& 	& 	& 	& 	& 	& \lstick{\ket{0}}			& \rygate{4}{6}	& \ctrl{-4}		& \ctrl{1}	& \targ{}	& \qw		& \qw			& \dsugate{1}{1}	& \qw 							\\
			\lstick{\ket{0}}	& \qw			& \qw			& \qw			& \qw		& \qw		& \targ{}		& \qw			& \qw 						& 	& 	& 	& 	& 	& \lstick{\ket{0}}			& \qw		& \qw			& \rygate{1}{2}	& \ctrl{-1}	& \qw		& \qw			& \dsugate{1}{1}	& \qw 							
		\end{quantikz}
	\end{adjustbox}
	\caption{Previous approaches to $\dicke{n}{k}$ preparation using parallel Dicke state unitaries $\dsu{k}{k}$:
		We use big-endian notation (top-to-bottom wires are right-to-left bitstrings) and shorthand gate notation 
		$\theta_{\surd x/y} := R_y(2\cos^{-1}(\sqrt{x/y}))\colon \ket{0} \mapsto \surd\tfrac{x}{y}\ket{0} + \surd\tfrac{y-x}{y}\ket{1}$.\newline
		(LEFT) Following~\cite{aktar2022divideconquer}, we prepare the superposition $\surd\tfrac{1}{20}\left(\surd1\ket{000} + \surd9\ket{001} + \surd9\ket{011} +\surd1\ket{111}\right)$
		on the least-significant qubits. Note that the numerators $1,9,9$ and the suffix sums $20,19,10$ appear as terms in the angle arguments. 
		Next, we add the missing Hamming weight to the other qubits, creating 
		$\smash{\tfrac{1}{\surd 20} \sum_{\ell=0}^3 \surd{\tbinom{3}{3-\ell}}\surd{\tbinom{3}{\ell}} \ket{0^{\ell}1^{3-\ell}}\ket{0^{3-\ell}1^{\ell}}}$.
		This initial state preparation is static, using a fixed Hamming weight $k=3$.
		Finally, we symmetrize the two qubit triples using Dicke state unitaries $\dsu{3}{3}$.	\newline
		(RIGHT) Following~\cite{epfl2019wstate}, we distribute an input Hamming weight of $0 \leq \ell \leq 1$ between the first four and the last two qubits,
		yielding $\smash{\surd\tfrac{2}{6}\ket{0(0^{1-\ell}1^{\ell})0000} + \surd\tfrac{4}{6}\ket{00000(0^{1-\ell}1^{\ell})}}$. Next, we recursively 
		apply this construction to further distribute the Hamming weight between all six qubits, yielding $\dicke{6}{\ell}$. 
		Crucially, the construction works for both an input of $\ket{000000}$ and of $\ket{000001}$, corresponding to $\ell=0$ and $\ell=1$, respectively.
		The Dicke state unitaries $\dsu{1}{1}$ are simply identities.
	}
	\label{fig:previous-methods}
\end{figure*}
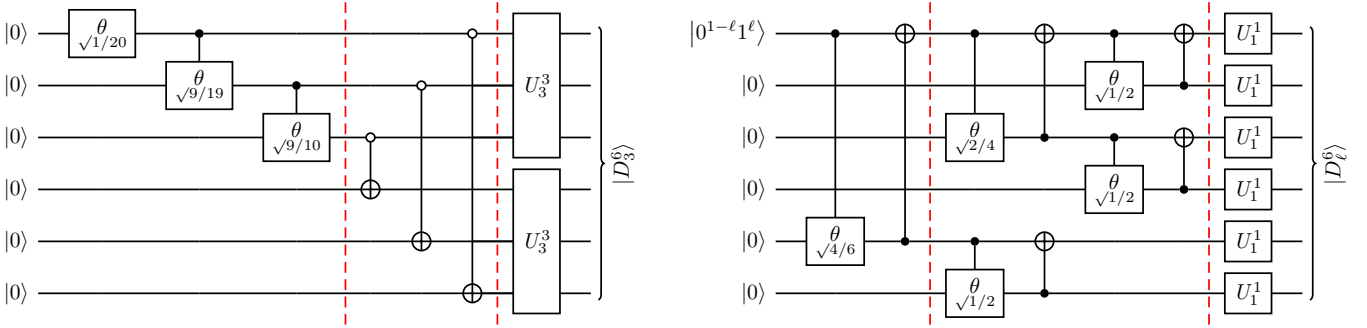

\section{Preliminaries and Main Idea}
\label{sec:preliminaries}

In order to prepare Dicke states, we reuse the concept of Dicke state unitaries~$\dsu{n}{k}$, which, given as input a classical state $\ket{0^{n-\ell}1^{\ell}}$ with $\ell \leq k$,
prepare the Dicke state $\dicke{n}{\ell}$. Note that the classical state $\ket{0^{n-\ell}1^{\ell}}$ is a 0-padded unary encoding of the Hamming weight $\ell$ of the Dicke state 
to be generated. We choose a big-endian bitstring representation (the most significant bit -- the ``big end'' -- on the left), where a top-to-bottom ordering of wires in circuit diagrams 
corresponds to right-to-left orderings in the ket notation:

\begin{definition}[adapted from~\cite{baertschi2019deterministic}]
	Denote by $\dsu{n}{k}$ any unitary satisfying for all $\ell \leq k\colon \dsu{n}{k} \ket{0^{n-\ell}1^{\ell}} = \dicke{n}{\ell}$.%
	\label{def:dsu}
\end{definition}

The use of such an input-dependent state preparation unitary has several benefits, with little overhead extending Dicke state preparation to Symmetric state preparation,
and -- using the inverse unitary $\dsu{n\dagger}{k}$ -- to Symmetric state compression:
\begin{itemize}
	\item	Symmetric states $\ket{\Psi}$ with non-zero amplitudes exclusively on computational basis states of Hamming weight $\leq k$ can be written in the Dicke state basis,
		$\sum_{\ell=0}^k \alpha_\ell \dicke{n}{\ell}$. 
		Thus they can be prepared by feeding the input $\sum_{\ell=0}^k \alpha_\ell \ket{0^{n-\ell}1^{\ell}}$ into $\dsu{n}{k}$. This input itself can be prepared in
		depth $\bigO(k)$, even on LNN connectivity~\cite{baertschi2019deterministic}.
	\item	An existing Symmetric state $\ket{\Psi}$ of this form can also be compressed into $\lceil \log(k+1) \rceil$ qubits, using the inverse unitary $\dsu{n\dagger}{k}$ to 
		get $\sum_{\ell=0}^k \alpha_\ell \ket{0^{n-\ell}1^{\ell}}$, followed by conversions into
		the one-hot-encoding $\sum_{\ell=0}^k \alpha_\ell \ket{0^{n-\ell}1^10^{\ell-1}}$ and then into
		a 0-padded logarithmic encoding $\sum_{\ell=0}^k \alpha_\ell \ket{0\ldots0}\ket{\ell}$~with only $\lceil \log(k+1) \rceil$ non-zero qubits~\cite{Plesch2010}. 
		These encoding transformations can be implemented in depth $\bigO(k \log k)$ for all-to-all connectivities and $\bigO(k\, \plog k)$ for LNN connectivities~\cite{baertschi2019deterministic}.
\end{itemize}

This allows us to exclusively focus on efficient implementations of the Dicke state unitary $\dsu{n}{k}$ for specific topologies. 
For LNN connectivity, we recall the following state-of-the art linear-depth result:

\begin{lemma}[paraphrased from~\cite{baertschi2019deterministic}]
	The Dicke state unitary $\dsu{n}{k}$ can be implemented with a circuit of depth $\bigO(n)$ on Linear Nearest Neighbor connectivities, 
	with $\bigO(k n)$ 2-qubit gates in total and without ancilla qubits.
	\label{lem:dsu-LNN}
\end{lemma}

\subsection{Main Idea} 
Informally speaking, we will use the circuit constructions behind Lemma~\ref{lem:dsu-LNN} 
by invoking $\bigO(\tfrac{n}{k})$ many parallelly executed instances of smaller Dicke state unitaries $\dsu{k}{k}$, each distributing some unary-encoded input Hamming weight $\ell \leq k$ across its $k$ qubits. 
This still needs $\bigO(\tfrac{n}{k} \cdot k^2) = \bigO(kn)$ gates in total, but at a depth of only $\bigO(k)$.

So what's the catch? The challenge is to provide a correctly weighted superposition of these inputs:
\begin{itemize}
	\item	The input Hamming weights $\ell_i$ fed to the different Dicke state unitaries $\dsu{k}{k}$ need to sum up to $k$, $\sum_{i=0}^{n/k} \ell_i \stackrel{!}{=} k$. 			
	\item	The amplitude of a specific Hamming weight combination input $(\ell_i)_i$ needs to reflect the number of computational basis states in the Dicke state resulting from $(\ell_i)_i$.
\end{itemize}

For illustration purposes, we give a concrete example, see also Figure~\ref{fig:previous-methods}~(left): The Dicke state $\dicke{6}{3}$ is a superposition of $\tbinom{6}{3}$ computational basis states of length $2\cdot 3$ with Hamming weight $3$. 
Note that there is only $1 = \tbinom{3}{0}\tbinom{3}{3}$ basis state that has Hamming weight 3 in the least-significant three bits: $\ket{000111}$.
However, there are $9 = \tbinom{3}{1}\tbinom{3}{2}$ basis states that have Hamming weight 2 in the least-significant three bits.
The same observations hold for Hamming weights 0 and 1, respectively.

Hence if we want to prepare $\dicke{6}{3}$ with a parallel execution of two Dicke state unitaries $\dsu{3}{3}$ on the three least- and the three most-significant qubits, 
we first have to prepare an input state that correctly distributes the Hamming weight 3 across these two sets of qubits:
\[	\dicke{6}{3} =	\left(\dsu{3}{3}\otimes \dsu{3}{3}\right) \frac{1}{\surd\tbinom{6}{3}} \sum_{\ell=0}^3 \sqrt{\tbinom{3}{3-\ell}\tbinom{3}{\ell}} \ket{0^{\ell}1^{3-\ell}}\ket{0^{3-\ell}1^{\ell}}.	\]

Indeed, preparing a Dicke state $\dicke{n}{k}$ using smaller Dicke state unitaries $\dsu{\lceil n/2 \rceil}{k}$ and $\dsu{\lfloor n/2 \rfloor}{k}$
has improved the fidelity of Dicke states prepared on IBMQ~\cite{aktar2022divideconquer}, as the resulting circuits lower the gate count.

\subsection{Hamming Weight Distribution}
In hindsight, previous methods can be discussed in terms of Hamming weight distribution.  
The approach above of splitting a fixed (non-input) Hamming weight $k$ across two sets of qubits does not lend itself to a recursive application, 
see Figure~\ref{fig:previous-methods}~(left) or an interactive Quirk~\cite{quirk} circuit for the 
\href{https://algassert.com/quirk#circuit=%7B%22cols%22%3A%5B%5B%7B%22id%22%3A%22Ryft%22%2C%22arg%22%3A%222%20acos(sqrt(1%2F20))%22%7D%5D%2C%5B%5D%2C%5B%5D%2C%5B%5D%2C%5B%22%E2%80%A2%22%2C%7B%22id%22%3A%22Ryft%22%2C%22arg%22%3A%222%20acos(sqrt(9%2F19))%22%7D%5D%2C%5B%5D%2C%5B%5D%2C%5B%5D%2C%5B1%2C%22%E2%80%A2%22%2C%7B%22id%22%3A%22Ryft%22%2C%22arg%22%3A%222%20acos(sqrt(9%2F10))%22%7D%5D%2C%5B%5D%2C%5B%5D%2C%5B%5D%2C%5B1%2C1%2C%22%E2%97%A6%22%2C%22X%22%5D%2C%5B1%2C%22%E2%97%A6%22%2C1%2C1%2C%22X%22%5D%2C%5B%22%E2%97%A6%22%2C1%2C1%2C1%2C1%2C%22X%22%5D%2C%5B%22~kfu0%22%2C1%2C1%2C%22~kfu0%22%5D%5D%2C%22gates%22%3A%5B%7B%22id%22%3A%22~kfu0%22%2C%22name%22%3A%22U3%2C3%22%2C%22circuit%22%3A%7B%22cols%22%3A%5B%5B1%2C%7B%22id%22%3A%22Ryft%22%2C%22arg%22%3A%22pi%2F2%22%7D%5D%2C%5B%22X%22%2C%22%E2%80%A2%22%5D%2C%5B%7B%22id%22%3A%22Ryft%22%2C%22arg%22%3A%22acos(sqrt(2%2F3))%22%7D%2C%7B%22id%22%3A%22Ryft%22%2C%22arg%22%3A%22acos(sqrt(2%2F3))%22%7D%5D%2C%5B%5D%2C%5B%5D%2C%5B%5D%2C%5B%22X%22%2C%22%E2%80%A2%22%5D%2C%5B1%2C%7B%22id%22%3A%22Ryft%22%2C%22arg%22%3A%22-pi%2F2%22%7D%5D%2C%5B%5D%2C%5B1%2C%22%E2%80%A2%22%2C%22X%22%5D%2C%5B1%2C%7B%22id%22%3A%22Ryft%22%2C%22arg%22%3A%22-0.5*acos(sqrt(2%2F3))%22%7D%5D%2C%5B%5D%2C%5B%5D%2C%5B%5D%2C%5B%5D%2C%5B%22%E2%80%A2%22%2C%22X%22%5D%2C%5B1%2C%7B%22id%22%3A%22Ryft%22%2C%22arg%22%3A%220.5*acos(sqrt(2%2F3))%22%7D%5D%2C%5B%5D%2C%5B%5D%2C%5B%5D%2C%5B1%2C%22X%22%2C%22%E2%80%A2%22%5D%2C%5B1%2C%7B%22id%22%3A%22Ryft%22%2C%22arg%22%3A%22-0.5*acos(sqrt(2%2F3))%22%7D%5D%2C%5B%5D%2C%5B%5D%2C%5B%5D%2C%5B%5D%2C%5B%22%E2%80%A2%22%2C%22X%22%5D%2C%5B1%2C%7B%22id%22%3A%22Ryft%22%2C%22arg%22%3A%220.5*acos(sqrt(2%2F3))%22%7D%5D%2C%5B%5D%2C%5B%5D%2C%5B%5D%2C%5B1%2C%22%E2%80%A2%22%2C%22X%22%5D%2C%5B1%2C%7B%22id%22%3A%22Ryft%22%2C%22arg%22%3A%22pi%2F2%22%7D%5D%2C%5B%22X%22%2C%22%E2%80%A2%22%5D%2C%5B%7B%22id%22%3A%22Ryft%22%2C%22arg%22%3A%22acos(sqrt(1%2F2))%22%7D%2C%7B%22id%22%3A%22Ryft%22%2C%22arg%22%3A%22acos(sqrt(1%2F2))%22%7D%5D%2C%5B%5D%2C%5B%5D%2C%5B%5D%2C%5B%22X%22%2C%22%E2%80%A2%22%5D%2C%5B1%2C%7B%22id%22%3A%22Ryft%22%2C%22arg%22%3A%22-pi%2F2%22%7D%5D%5D%7D%7D%5D%7D
}{Dicke state \dicke{6}{3}}.
Hence there is no asymptotic improvement of the state preparation depth over a single $\dsu{n}{k}$ unitary~\cite{aktar2022divideconquer}.
On the other hand, an asymptotic improvement exists for $W$ states on all-to-all connectivities~\cite{epfl2019wstate} (note that $\dsu{1}{1} = \Id$). 
It is based on a recursive application of a subroutine which can correctly redistribute both Hamming weight 0 and Hamming weight 1, but does not generalize to higher Hamming weights,
see Figure~\ref{fig:previous-methods}~(right) or a Quirk circuit for 
\href{https://algassert.com/quirk#circuit=%7B%22cols%22%3A%5B%5B%22Counting1%22%5D%2C%5B%22Chance%22%5D%2C%5B%22%E2%80%A2%22%2C1%2C1%2C1%2C%7B%22id%22%3A%22Ryft%22%2C%22arg%22%3A%222%20acos(sqrt(4%2F6))%22%7D%5D%2C%5B%5D%2C%5B%5D%2C%5B%5D%2C%5B%22X%22%2C1%2C1%2C1%2C%22%E2%80%A2%22%5D%2C%5B%22Chance6%22%5D%2C%5B%22%E2%80%A2%22%2C1%2C%7B%22id%22%3A%22Ryft%22%2C%22arg%22%3A%222%20acos(sqrt(2%2F4))%22%7D%5D%2C%5B1%2C1%2C1%2C1%2C%22%E2%80%A2%22%2C%7B%22id%22%3A%22Ryft%22%2C%22arg%22%3A%222%20acos(sqrt(1%2F2))%22%7D%5D%2C%5B%5D%2C%5B%5D%2C%5B%22X%22%2C1%2C%22%E2%80%A2%22%5D%2C%5B1%2C1%2C1%2C1%2C%22X%22%2C%22%E2%80%A2%22%5D%2C%5B%22%E2%80%A2%22%2C%7B%22id%22%3A%22Ryft%22%2C%22arg%22%3A%222%20acos(sqrt(1%2F2))%22%7D%5D%2C%5B1%2C1%2C%22%E2%80%A2%22%2C%7B%22id%22%3A%22Ryft%22%2C%22arg%22%3A%222%20acos(sqrt(1%2F2))%22%7D%5D%2C%5B%5D%2C%5B%5D%2C%5B%5D%2C%5B%22X%22%2C%22%E2%80%A2%22%5D%2C%5B1%2C1%2C%22X%22%2C%22%E2%80%A2%22%5D%5D%7D
}{Dicke states \dicke{6}{0,1}}.

Our approach is therefore to augment both of these methods, such that we get a subroutine which can distribute any Hamming weight $0 \leq \ell \leq k$ (given as a unary encoding) on $n>k$ qubits 
into a partition of $m$ and $n-m$ qubits. Such subroutines can then recursively be called on these smaller sets of $m$ and $n-m$ qubits if $m>k$ and $n-m>k$ respectively. 
Thus they build the ``Divide'' part of a recursive divide-and-conquer approach.  
Otherwise, we are already at the last step and can distribute unary Hamming weights symmetrically across the qubits using Dicke state unitaries $\dsu{m}{m}$ and/or $\dsu{n-m}{n-m}$, respectively.
These Dicke state unitaries thus build the ``Conquer'' part of a recursive divide-and-conquer approach. 
We define a corresponding \emph{weight distribution block}:

\begin{definition}[Weight Distribution Block]
	Denote by $\wdb{n}{m}{k}$ any unitary mapping for all $\ell \leq k\colon \ket{0^{n-\ell}1^\ell}$
	\[ \mapsto \frac{1}{\surd\tbinom{n}{\ell}} \sum_{i=0}^{\ell} \sqrt{\tbinom{m}{i}\tbinom{n-m}{\ell-i}} \ket{0^{m-i}1^i}\ket{0^{n-m+i-\ell}1^{\ell-i}},	\]
	where we set $\tbinom{x}{y} = 0$ whenever $y>x$ (which can happen, e.g., if $m<k$ for terms with $\ell>m$). 
	\label{def:wdb}
\end{definition}

\begin{figure}[t!]
	\centering
	\begin{adjustbox}{width=\linewidth}
	%% Grover Mixer QAOA
	%% -----------------
	%\hspace*{-0.5cm}
		\begin{quantikz}[row sep={24pt,between origins},execute at end picture={}]
			\lstick[3]{\rotatebox{90}{$\ket{0^{3-\ell}1^{\ell}}$}}	& \gate[9,nwires={4,5,6}]{\wdb{11}{5}{3}}	& \qw			& \qw			& \wdbgate{6}{3}{3}	& \qw	& \qw	& \dsugate{3}{3}	& \qw\rstick[11]{\rotatebox{90}{$\dicke{11}{\ell}$}} 	\\ 
			\							&						& \qw			& \qw			& 			& \qw	& \qw	& \qw			& \qw					 		\\
			\							&						& \qw			& \qw			& 			& \qw	& \qw	& \qw			& \qw						\\
			\							&						& 			& \lstick{\ket{0}}	& 			& \qw	& \qw	& \dsugate{3}{3}	& \qw						\\
			\							&						& 			& \lstick{\ket{0}}	& 			& \qw	& \qw	& \qw			& \qw						\\
			\							& 						& 			& \lstick{\ket{0}}	& 			& \qw	& \qw	& \qw			& \qw						\\
			\lstick{\ket{0}}					& \qw						& \qw			& \qw			& \wdbgate{5}{2}{3}	& \qw	& \qw	& \dsugate{3}{3}	& \qw						\\
			\lstick{\ket{0}}					& \qw						& \qw			& \qw			& 			& \qw	& \qw	& \qw			& \qw						\\
			\lstick{\ket{0}}					& \qw						& \qw			& \qw			& 			& \qw	& \qw	& \qw			& \qw						\\
			\							&						& 			& \lstick{\ket{0}}	& 			& \qw	& \qw	& \dsugate{2}{2}	& \qw						\\
			\							&						& 			& \lstick{\ket{0}}	& 			& \qw	& \qw	& \qw			& \qw						
		\end{quantikz}
	\end{adjustbox}
	\caption{Our method: An input Hamming weight $0 \leq \ell \leq 3$ in a unary encoding $\ket{0^8 0^{3-\ell} 1^{\ell}}$ for 11 qubits
		is distributed between the bottom $5$ and the top $11-5=6$ qubits by $\wdb{11}{5}{3}$.\newline
		It is then recursively distributed on the top 6 qubits between $3$ and $6-3=3$ qubits using $\wdb{6}{3}{3}$,
		and on the bottom 5 qubits between $2$ and $5-2=3$ qubits using $\wdb{6}{3}{3}$.\newline
		Finally, the Hamming weights on the four sets of $3,3,3$ and $2$ qubits (still in unary encodings) are symmetrically distributed
		by Dicke state unitaries $\dsu{3}{3}, \dsu{3}{3}, \dsu{3}{3}$ and $\dsu{2}{2}$, respectively.\\[-2em]
	}
	\label{fig:our-method}
\end{figure}
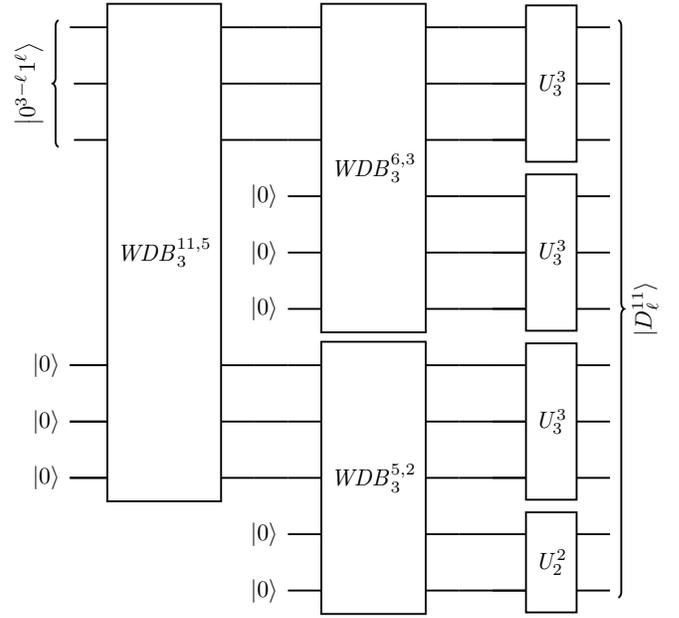

Note that such $\wdb{n}{m}{k}$ acts trivially on all but the $k$ least-significant qubits of both sets of $m$ and $n-m$ qubits.
In Figure~\ref{fig:our-method} we give a complete overview of how Dicke states $\dicke{11}{\ell}$ for $0\leq \ell \leq 3$ can be prepared 
using weight distribution blocks $\wdb{11}{5}{3}, \wdb{6}{3}{3}, \wdb{5}{2}{3}$ and Dicke state unitaries $\dsu{3}{3}$ and $\dsu{2}{2}$.
An interactive Quirk circuit for the whole construction can be found following 
\href{https://algassert.com/quirk#circuit=%7B%22cols%22%3A%5B%5B%22~jat7%22%5D%2C%5B%22Chance%22%2C%22Chance%22%2C%22Chance%22%5D%2C%5B%22Chance3%22%2C1%2C1%2C1%2C1%2C1%2C%22~m3kg%22%2C%22~m3kg%22%2C%22~m3kg%22%5D%2C%5B%22X%22%2C%22%E2%80%A2%22%5D%2C%5B1%2C%22X%22%2C%22%E2%80%A2%22%5D%2C%5B1%2C1%2C%22%E2%80%A2%22%2C1%2C1%2C1%2C%22~3t19%22%5D%2C%5B1%2C1%2C%22%E2%80%A2%22%2C1%2C1%2C1%2C%22%E2%80%A2%22%2C%22~dh7n%22%5D%2C%5B1%2C1%2C%22%E2%80%A2%22%2C1%2C1%2C1%2C1%2C%22%E2%80%A2%22%2C%22~cq6s%22%5D%2C%5B1%2C%22%E2%80%A2%22%2C1%2C1%2C1%2C1%2C%22~m5i%22%5D%2C%5B1%2C%22%E2%80%A2%22%2C1%2C1%2C1%2C1%2C%22%E2%80%A2%22%2C%22~i8ql%22%5D%2C%5B%22%E2%80%A2%22%2C1%2C1%2C1%2C1%2C1%2C%22~kcer%22%5D%2C%5B1%2C%22X%22%2C%22%E2%80%A2%22%5D%2C%5B%22X%22%2C%22%E2%80%A2%22%5D%2C%5B1%2C1%2C%22X%22%2C1%2C1%2C1%2C1%2C1%2C%22%E2%80%A2%22%5D%2C%5B1%2C%22%3E%3E2%22%2C1%2C1%2C1%2C1%2C1%2C%22%E2%80%A2%22%5D%2C%5B1%2C1%2C%22X%22%2C1%2C1%2C1%2C1%2C%22%E2%80%A2%22%5D%2C%5B%22%3E%3E3%22%2C1%2C1%2C1%2C1%2C1%2C%22%E2%80%A2%22%5D%2C%5B1%2C1%2C%22X%22%2C1%2C1%2C1%2C%22%E2%80%A2%22%5D%2C%5B%22Chance3%22%2C1%2C1%2C1%2C1%2C1%2C%22Chance3%22%5D%2C%5B%22%E2%80%A6%22%5D%2C%5B%22Chance3%22%2C1%2C1%2C%22~m3kg%22%2C%22~m3kg%22%2C%22~m3kg%22%5D%2C%5B%22X%22%2C%22%E2%80%A2%22%5D%2C%5B1%2C%22X%22%2C%22%E2%80%A2%22%5D%2C%5B1%2C1%2C%22%E2%80%A2%22%2C%22~2bbd%22%5D%2C%5B1%2C1%2C%22%E2%80%A2%22%2C%22%E2%80%A2%22%2C%22~e5cj%22%5D%2C%5B1%2C1%2C%22%E2%80%A2%22%2C1%2C%22%E2%80%A2%22%2C%22~i81f%22%5D%2C%5B1%2C%22%E2%80%A2%22%2C1%2C%22~ppan%22%5D%2C%5B1%2C%22%E2%80%A2%22%2C1%2C%22%E2%80%A2%22%2C%22~2mea%22%5D%2C%5B%22%E2%80%A2%22%2C1%2C1%2C%22~gv9n%22%5D%2C%5B1%2C%22X%22%2C%22%E2%80%A2%22%5D%2C%5B%22X%22%2C%22%E2%80%A2%22%5D%2C%5B1%2C1%2C%22X%22%2C1%2C1%2C%22%E2%80%A2%22%5D%2C%5B1%2C%22%3E%3E2%22%2C1%2C1%2C%22%E2%80%A2%22%5D%2C%5B1%2C1%2C%22X%22%2C1%2C%22%E2%80%A2%22%5D%2C%5B%22%3E%3E3%22%2C1%2C1%2C%22%E2%80%A2%22%5D%2C%5B1%2C1%2C%22X%22%2C%22%E2%80%A2%22%5D%2C%5B%22Chance3%22%2C1%2C1%2C%22Chance3%22%5D%2C%5B%22%E2%80%A6%22%5D%2C%5B1%2C1%2C1%2C1%2C1%2C1%2C%22Chance3%22%2C1%2C1%2C%22~m3kg%22%2C%22~m3kg%22%5D%2C%5B1%2C1%2C1%2C1%2C1%2C1%2C%22X%22%2C%22%E2%80%A2%22%5D%2C%5B1%2C1%2C1%2C1%2C1%2C1%2C1%2C%22X%22%2C%22%E2%80%A2%22%5D%2C%5B1%2C1%2C1%2C1%2C1%2C1%2C1%2C1%2C%22%E2%80%A2%22%2C%22~o2im%22%5D%2C%5B1%2C1%2C1%2C1%2C1%2C1%2C1%2C1%2C%22%E2%80%A2%22%2C%22%E2%80%A2%22%2C%22~9u45%22%5D%2C%5B1%2C1%2C1%2C1%2C1%2C1%2C1%2C%22%E2%80%A2%22%2C1%2C%22~v2k1%22%5D%2C%5B1%2C1%2C1%2C1%2C1%2C1%2C1%2C%22%E2%80%A2%22%2C1%2C%22%E2%80%A2%22%2C%22~4npg%22%5D%2C%5B1%2C1%2C1%2C1%2C1%2C1%2C%22%E2%80%A2%22%2C1%2C1%2C%22~l3k5%22%5D%2C%5B1%2C1%2C1%2C1%2C1%2C1%2C1%2C%22X%22%2C%22%E2%80%A2%22%5D%2C%5B1%2C1%2C1%2C1%2C1%2C1%2C%22X%22%2C%22%E2%80%A2%22%5D%2C%5B1%2C1%2C1%2C1%2C1%2C1%2C1%2C%22%3E%3E2%22%2C1%2C1%2C%22%E2%80%A2%22%5D%2C%5B1%2C1%2C1%2C1%2C1%2C1%2C1%2C1%2C%22X%22%2C1%2C%22%E2%80%A2%22%5D%2C%5B1%2C1%2C1%2C1%2C1%2C1%2C%22%3E%3E3%22%2C1%2C1%2C%22%E2%80%A2%22%5D%2C%5B1%2C1%2C1%2C1%2C1%2C1%2C1%2C1%2C%22X%22%2C%22%E2%80%A2%22%5D%2C%5B1%2C1%2C1%2C1%2C1%2C1%2C%22Chance3%22%2C1%2C1%2C%22Chance2%22%5D%2C%5B%22~30he%22%2C1%2C1%2C%22~30he%22%2C1%2C1%2C%22~30he%22%2C1%2C1%2C%22~j9or%22%5D%5D%2C%22gates%22%3A%5B%7B%22id%22%3A%22~6dbg%22%2C%22name%22%3A%22%CE%B8%5Cn%E2%88%9A1%2F2%22%2C%22circuit%22%3A%7B%22cols%22%3A%5B%5B%7B%22id%22%3A%22Ryft%22%2C%22arg%22%3A%222acos(sqrt(1%2F2))%22%7D%5D%5D%7D%7D%2C%7B%22id%22%3A%22~16b9%22%2C%22name%22%3A%22%CE%B8%5Cn%E2%88%9A1%2F3%22%2C%22circuit%22%3A%7B%22cols%22%3A%5B%5B%7B%22id%22%3A%22Ryft%22%2C%22arg%22%3A%222acos(sqrt(1%2F3))%22%7D%5D%5D%7D%7D%2C%7B%22id%22%3A%22~qdji%22%2C%22name%22%3A%22%CE%B8%5Cn%E2%88%9A2%2F3%22%2C%22circuit%22%3A%7B%22cols%22%3A%5B%5B%7B%22id%22%3A%22Ryft%22%2C%22arg%22%3A%222acos(sqrt(2%2F3))%22%7D%5D%5D%7D%7D%2C%7B%22id%22%3A%22~jat7%22%2C%22name%22%3A%22Test%22%2C%22circuit%22%3A%7B%22cols%22%3A%5B%5B%22Counting2%22%5D%2C%5B%22%E2%80%A2%22%2C%22%E2%80%A2%22%2C%22X%22%5D%2C%5B%22X%22%2C%22%E2%80%A2%22%5D%2C%5B%22X%22%2C1%2C%22%E2%80%A2%22%5D%5D%7D%7D%2C%7B%22id%22%3A%22~j9or%22%2C%22name%22%3A%22U2%2C2%22%2C%22circuit%22%3A%7B%22cols%22%3A%5B%5B%22X%22%2C%22%E2%80%A2%22%5D%2C%5B%22%E2%80%A2%22%2C%22~6dbg%22%5D%2C%5B%22X%22%2C%22%E2%80%A2%22%5D%5D%7D%7D%2C%7B%22id%22%3A%22~30he%22%2C%22name%22%3A%22U3%2C3%22%2C%22circuit%22%3A%7B%22cols%22%3A%5B%5B%22X%22%2C%22%E2%80%A2%22%5D%2C%5B%22%E2%80%A2%22%2C%22~16b9%22%5D%2C%5B%22X%22%2C%22%E2%80%A2%22%5D%2C%5B%22X%22%2C1%2C%22%E2%80%A2%22%5D%2C%5B%22%E2%80%A2%22%2C%22%E2%80%A2%22%2C%22~qdji%22%5D%2C%5B%22X%22%2C1%2C%22%E2%80%A2%22%5D%2C%5B1%2C%22X%22%2C%22%E2%80%A2%22%5D%2C%5B1%2C%22%E2%80%A2%22%2C%22~6dbg%22%5D%2C%5B1%2C%22X%22%2C%22%E2%80%A2%22%5D%5D%7D%7D%2C%7B%22id%22%3A%22~3t19%22%2C%22name%22%3A%22%CE%B8%5Cn%E2%88%9A20%2F165%22%2C%22circuit%22%3A%7B%22cols%22%3A%5B%5B%7B%22id%22%3A%22Ryft%22%2C%22arg%22%3A%222acos(sqrt(20%2F165))%22%7D%5D%5D%7D%7D%2C%7B%22id%22%3A%22~dh7n%22%2C%22name%22%3A%22%CE%B8%5Cn%E2%88%9A75%2F145%22%2C%22circuit%22%3A%7B%22cols%22%3A%5B%5B%7B%22id%22%3A%22Ryft%22%2C%22arg%22%3A%222acos(sqrt(75%2F145))%22%7D%5D%5D%7D%7D%2C%7B%22id%22%3A%22~cq6s%22%2C%22name%22%3A%22%CE%B8%5Cn%E2%88%9A60%2F70%22%2C%22circuit%22%3A%7B%22cols%22%3A%5B%5B%7B%22id%22%3A%22Ryft%22%2C%22arg%22%3A%222acos(sqrt(60%2F70))%22%7D%5D%5D%7D%7D%2C%7B%22id%22%3A%22~m5i%22%2C%22name%22%3A%22%CE%B8%5Cn%E2%88%9A15%2F55%22%2C%22circuit%22%3A%7B%22cols%22%3A%5B%5B%7B%22id%22%3A%22Ryft%22%2C%22arg%22%3A%222acos(sqrt(15%2F55))%22%7D%5D%5D%7D%7D%2C%7B%22id%22%3A%22~i8ql%22%2C%22name%22%3A%22%CE%B8%5Cn%E2%88%9A30%2F40%22%2C%22circuit%22%3A%7B%22cols%22%3A%5B%5B%7B%22id%22%3A%22Ryft%22%2C%22arg%22%3A%222acos(sqrt(30%2F40))%22%7D%5D%5D%7D%7D%2C%7B%22id%22%3A%22~kcer%22%2C%22name%22%3A%22%CE%B8%5Cn%E2%88%9A6%2F11%22%2C%22circuit%22%3A%7B%22cols%22%3A%5B%5B%7B%22id%22%3A%22Ryft%22%2C%22arg%22%3A%222acos(sqrt(6%2F11))%22%7D%5D%5D%7D%7D%2C%7B%22id%22%3A%22~2bbd%22%2C%22name%22%3A%22%CE%B8%5Cn%E2%88%9A1%2F20%22%2C%22circuit%22%3A%7B%22cols%22%3A%5B%5B%7B%22id%22%3A%22Ryft%22%2C%22arg%22%3A%222acos(sqrt(1%2F20))%22%7D%5D%5D%7D%7D%2C%7B%22id%22%3A%22~e5cj%22%2C%22name%22%3A%22%CE%B8%5Cn%E2%88%9A9%2F19%22%2C%22circuit%22%3A%7B%22cols%22%3A%5B%5B%7B%22id%22%3A%22Ryft%22%2C%22arg%22%3A%222acos(sqrt(9%2F19))%22%7D%5D%5D%7D%7D%2C%7B%22id%22%3A%22~i81f%22%2C%22name%22%3A%22%CE%B8%5Cn%E2%88%9A9%2F10%22%2C%22circuit%22%3A%7B%22cols%22%3A%5B%5B%7B%22id%22%3A%22Ryft%22%2C%22arg%22%3A%222acos(sqrt(9%2F10))%22%7D%5D%5D%7D%7D%2C%7B%22id%22%3A%22~ppan%22%2C%22name%22%3A%22%CE%B8%5Cn%E2%88%9A3%2F15%22%2C%22circuit%22%3A%7B%22cols%22%3A%5B%5B%7B%22id%22%3A%22Ryft%22%2C%22arg%22%3A%222acos(sqrt(3%2F15))%22%7D%5D%5D%7D%7D%2C%7B%22id%22%3A%22~2mea%22%2C%22name%22%3A%22%CE%B8%5Cn%E2%88%9A9%2F12%22%2C%22circuit%22%3A%7B%22cols%22%3A%5B%5B%7B%22id%22%3A%22Ryft%22%2C%22arg%22%3A%222acos(sqrt(9%2F12))%22%7D%5D%5D%7D%7D%2C%7B%22id%22%3A%22~gv9n%22%2C%22name%22%3A%22%CE%B8%5Cn%E2%88%9A3%2F6%22%2C%22circuit%22%3A%7B%22cols%22%3A%5B%5B%7B%22id%22%3A%22Ryft%22%2C%22arg%22%3A%222acos(sqrt(3%2F6))%22%7D%5D%5D%7D%7D%2C%7B%22id%22%3A%22~o2im%22%2C%22name%22%3A%22%CE%B8%5Cn%E2%88%9A1%2F10%22%2C%22circuit%22%3A%7B%22cols%22%3A%5B%5B%7B%22id%22%3A%22Ryft%22%2C%22arg%22%3A%222acos(sqrt(1%2F10))%22%7D%5D%5D%7D%7D%2C%7B%22id%22%3A%22~9u45%22%2C%22name%22%3A%22%CE%B8%5Cn%E2%88%9A6%2F9%22%2C%22circuit%22%3A%7B%22cols%22%3A%5B%5B%7B%22id%22%3A%22Ryft%22%2C%22arg%22%3A%222acos(sqrt(6%2F9))%22%7D%5D%5D%7D%7D%2C%7B%22id%22%3A%22~v2k1%22%2C%22name%22%3A%22%CE%B8%5Cn%E2%88%9A3%2F10%22%2C%22circuit%22%3A%7B%22cols%22%3A%5B%5B%7B%22id%22%3A%22Ryft%22%2C%22arg%22%3A%222acos(sqrt(3%2F10))%22%7D%5D%5D%7D%7D%2C%7B%22id%22%3A%22~4npg%22%2C%22name%22%3A%22%CE%B8%5Cn%E2%88%9A6%2F7%22%2C%22circuit%22%3A%7B%22cols%22%3A%5B%5B%7B%22id%22%3A%22Ryft%22%2C%22arg%22%3A%222acos(sqrt(6%2F7))%22%7D%5D%5D%7D%7D%2C%7B%22id%22%3A%22~l3k5%22%2C%22name%22%3A%22%CE%B8%5Cn%E2%88%9A3%2F5%22%2C%22circuit%22%3A%7B%22cols%22%3A%5B%5B%7B%22id%22%3A%22Ryft%22%2C%22arg%22%3A%222acos(sqrt(3%2F5))%22%7D%5D%5D%7D%7D%2C%7B%22id%22%3A%22~m3kg%22%2C%22name%22%3A%22%7C0%E2%9F%A9%22%2C%22matrix%22%3A%22%7B%7B1%2C0%7D%2C%7B0%2C1%7D%7D%22%7D%5D%7D
}{this link}.

\paragraph*{Comparison}
It is worth to interpret previous work in this new terminology: 
The logarithmic depth scheme for $W$ states~\cite{epfl2019wstate} depicted in Figure~\ref{fig:previous-methods}~(right) uses weight distribution blocks $\wdb{6}{2}{1}, \wdb{4}{2}{1}, \wdb{2}{1}{1}$.
In general, it uses $\wdb{n}{m}{1}$ for arbitrary $n,m$ but restricted to $k=1$.

The scheme for preparing Dicke states $\dicke{n}{m}$ using two Dicke state unitaries~\cite{aktar2022divideconquer} $\dsu{m}{k}, \dsu{n-m}{k}$ shown in Figure~\ref{fig:previous-methods}~(left) can be interpreted as
a single weight distribution block $\wdb{n}{m}{k}$ that \emph{only accepts a fixed input} $\ell = k$, but not $\ell< k$. 

The scheme behind the result in Lemma~\ref{lem:dsu-LNN} uses a weight distribution block $\wdb{n}{1}{k}$ (called a split \& cyclic shift $\SCS$ gate in~\cite{baertschi2019deterministic}), 
which an be implemented in $\bigO(n)$ depth on LNN topologies using a ``stair-shaped'' construction. These stair shapes are the reason, why recursive applications of these $\SCS$ gates 
result in a total depth of $\bigO(n)$ rather than $\bigO(n^2)$. 

\paragraph*{Novelty}
Our construction extends and unifies both the theoretical insights as well as the constructions behind all these approaches to a fully general weight distribution block $\wdb{n}{m}{k}$. 
A crucial observation is that $\wdb{n}{m}{k}$ acts trivially on all but the last $k$ qubits of each set of $m$ and $n-m$ qubits.
For the topology-dependent structures for the recursive application of weight-distribution blocks we will therefore carefully arrange the qubits such that these two sets of 
$k$ qubits each are neighbors along a LNN connectivity spanning the $2k$ qubits. This brings us to the ``Divide'' part of our approach.

%% Weight Distribution Block
\begin{figure*}[t!]
	\centering
	\begin{adjustbox}{width=\linewidth}
	%% Grover Mixer QAOA
	%% -----------------
	%\hspace*{-0.5cm}
		\def\sp{24pt}		
		\begin{quantikz}[row sep={\sp,between origins},execute at end picture={
					%% parallelograms
					\draw[thick,fill=white] ($(\tikzcdmatrixname-5-6)+(-18pt,12pt)$) -- ($(\tikzcdmatrixname-5-6)+(15pt,12pt)$) -- ($(\tikzcdmatrixname-5-6)+(15pt,12pt)+(4*\sp,-4*\sp)$) -- ($(\tikzcdmatrixname-5-6)+(-18pt,12pt)+(4*\sp,-4*\sp)$) -- cycle;
					\draw[thick,fill=white] ($(\tikzcdmatrixname-5-9)+(-18pt,12pt)$) -- ($(\tikzcdmatrixname-5-9)+(15pt,12pt)$) -- ($(\tikzcdmatrixname-5-9)+(15pt,12pt)+(3*\sp,-3*\sp)$) -- ($(\tikzcdmatrixname-5-9)+(-18pt,12pt)+(3*\sp,-3*\sp)$) -- cycle;
					\draw[thick,fill=white] ($(\tikzcdmatrixname-5-12)+(-18pt,12pt)$) -- ($(\tikzcdmatrixname-5-12)+(15pt,12pt)$) -- ($(\tikzcdmatrixname-5-12)+(15pt,12pt)+(2*\sp,-2*\sp)$) -- ($(\tikzcdmatrixname-5-12)+(-18pt,12pt)+(2*\sp,-2*\sp)$) -- cycle;
					\draw[thick,fill=white] ($(\tikzcdmatrixname-5-15)+(-18pt,12pt)$) -- ($(\tikzcdmatrixname-5-15)+(15pt,12pt)$) -- ($(\tikzcdmatrixname-5-15)+(15pt,12pt)+(\sp,-\sp)$) -- ($(\tikzcdmatrixname-5-15)+(-18pt,12pt)+(\sp,-\sp)$) -- cycle;
					%% permutations
					% 2-bit
					\draw[thick] ($(\tikzcdmatrixname-3-23)-(15pt,0pt)$) to ($(\tikzcdmatrixname-4-23)+(15pt,0pt)$);					
					\draw[thick] ($(\tikzcdmatrixname-4-23)-(15pt,0pt)$) to ($(\tikzcdmatrixname-3-23)+(15pt,0pt)$);
					% 3-bit
					\draw[thick] ($(\tikzcdmatrixname-2-25)-(15pt,0pt)$) to ($(\tikzcdmatrixname-4-25)+(15pt,0pt)$);					
					\draw[thick] ($(\tikzcdmatrixname-3-25)-(15pt,0pt)$) to ($(\tikzcdmatrixname-2-25)+(15pt,0pt)$);					
					\draw[thick] ($(\tikzcdmatrixname-4-25)-(15pt,0pt)$) to ($(\tikzcdmatrixname-3-25)+(15pt,0pt)$);					
					% 4-bit
					\draw[thick] ($(\tikzcdmatrixname-1-27)-(15pt,0pt)$) to ($(\tikzcdmatrixname-4-27)+(15pt,0pt)$);										
					\draw[thick] ($(\tikzcdmatrixname-2-27)-(15pt,0pt)$) to ($(\tikzcdmatrixname-1-27)+(15pt,0pt)$);										
					\draw[thick] ($(\tikzcdmatrixname-3-27)-(15pt,0pt)$) to ($(\tikzcdmatrixname-2-27)+(15pt,0pt)$);										
					\draw[thick] ($(\tikzcdmatrixname-4-27)-(15pt,0pt)$) to ($(\tikzcdmatrixname-3-27)+(15pt,0pt)$);										
			}]
			\lstick[4]{\rotatebox{90}{\large$\ket{0^{4-\ell}1^{\ell}}$}}	& \targ{}	& \qw		& \qw\slice{(1)}& \qw	& \qw		& \qw	& \qw	& \qw		& \qw	& \qw	& \qw		& \qw	& \qw	& \ctrl{4}	& \qw		& \qw		& \qw\slice{(2)}& \qw		& \qw		& \targ{}  	& \qw		& \qw		& \qw		& \qw		& \qw		& \gate[4]{}WW	& \qw\slice{(4)}& \qw\rstick[8]{\rotatebox{90}{\large$\frac{1}{\surd\binom{n}{\ell}} \sum\limits_{i=0}^{\ell} \sqrt{\tbinom{m}{i}\tbinom{n-m}{\ell-i}} \ket{0^{4-i}1^i} \ket{0^{4+i-\ell}1^{\ell-i}} $}}		\\ 
			\								& \ctrl{-1}	& \targ{}  	& \qw		& \qw	& \qw		& \qw	& \qw	& \qw		& \qw	& \qw	& \ctrl{3}	& \qw	& \qw	& \qw		& \qw		& \qw		& \qw		& \qw		& \targ{}  	& \ctrl{-1}	& \qw		& \qw		& \qw		& \gate[3]{}WW	& \qw		& \qw		& \qw		& \qw	\\
			\								& \qw		& \ctrl{-1}	& \targ{}  	& \qw	& \qw		& \qw	& \qw	& \ctrl{2}	& \qw	& \qw	& \qw		& \qw	& \qw	& \qw		& \qw		& \qw		& \qw		& \targ{}  	& \ctrl{-1}	& \qw\slice{(3)}& \qw		& \gate[2]{}WW	& \qw		& \qw		& \qw		& \qw		& \qw		& \qw	\\
			\								& \qw		& \qw		& \ctrl{-1}	& \qw	& \ctrl{1}	& \qw	& \qw	& \qw		& \qw	& \qw	& \qw		& \qw	& \qw	& \qw		& \qw		& \qw		& \qw		& \ctrl{-1}	& \qw		& \qw		& \targ{}	& \qw		& \targ{}	& \qw		& \targ{}	& \qw		& \targ{}	& \qw	\\
			\lstick{\ket{0}}						& \qw		& \qw		& \qw		& \qw	& \qw		& \qw	& \qw	& \qw		& \qw	& \qw	& \qw		& \qw	& \qw	& \qw		& \qw		& \qw		& \qw		& \qw		& \qw		& \qw		& \qw		& \qw		& \qw		& \qw		& \qw		& \ctrl{-1}	& \ctrl{-1}	& \qw	\\
			\lstick{\ket{0}}						& \qw		& \qw		& \qw		& \qw	& \qw		& \qw	& \qw	& \qw		& \qw	& \qw	& \qw		& \qw	& \qw	& \qw		& \qw		& \qw		& \qw		& \qw		& \qw		& \qw		& \qw		& \qw		& \qw		& \ctrl{-2}	& \ctrl{-2}	& \qw		& \qw		& \qw	\\
			\lstick{\ket{0}}						& \qw		& \qw		& \qw		& \qw	& \qw		& \qw	& \qw	& \qw		& \qw	& \qw	& \qw		& \qw	& \qw	& \qw		& \qw		& \qw		& \qw		& \qw		& \qw		& \qw		& \qw		& \ctrl{-3}	& \ctrl{-3}	& \qw		& \qw		& \qw		& \qw		& \qw	\\	
			\lstick{\ket{0}}						& \qw		& \qw		& \qw		& \qw	& \qw		& \qw	& \qw	& \qw		& \qw	& \qw	& \qw		& \qw	& \qw	& \qw		& \qw		& \qw		& \qw		& \qw		& \qw		& \qw		& \ctrl{-4}	& \qw		& \qw		& \qw		& \qw		& \qw		& \qw		& \qw		
		\end{quantikz}
	\end{adjustbox}
	\caption{Schematics of a weight distribution block $\wdb{n}{m}{k=4}$, where we restrict ourselves to show the relevant $4+4$ of the $n$ qubits. We have from left to right: 
			(1) switching from a unary encoding of the input Hamming weight $\ell$ to a one-hot-encoding;
			(2) adding a superposition of Hamming weights $i \leq \ell$ into the second register, using a unary encoding;
			(3) switching back to a unary encoding in the first register;
			(4) subtracting the weight $i$ in the second register from the weight $\ell$ in the first register.\newline
			The large controlled gates in steps 2, 4 have depth $\bigO(k)$ but also a stair shape. Thus they can be executed partially in parallel, increasing the asymptotic depth by only a constant factor.
		If $m<4$, the second register has to be truncated to $m$ qubits.
	}
	\label{fig:wdb}
\end{figure*}
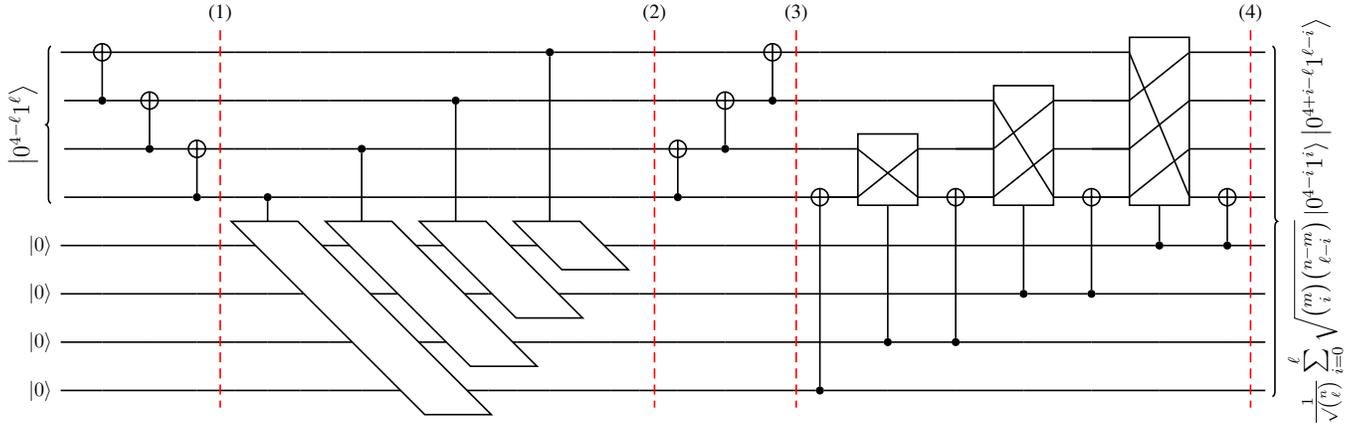

%% Weight Distribution Block
\section{Weight Distribution Block Circuits}
\label{sec:divide}

In this section, we present a circuit construction for a weight distribution block $\wdb{n}{m}{k}$ which has depth $\bigO(k)$ and needs $\bigO(k^2)$ gates but no ancilla qubits.
Its schematic is given in Figure~\ref{fig:wdb} and details are presented in Figure~\ref{fig:add-subtract}. 
Recall from Definition~\ref{def:wdb}, that for any unary encoded input $0\leq \ell \leq k$ for $n$ qubits, $\wdb{n}{m}{k}$ has to distribute $\ell$ into superpositions
of (correctly weighted) unary encoded inputs into $n-m$ and $m$ many qubits. As $\ell \leq k$, we can restrict ourselves to a circuit involving the first $k$ qubits of the latter sets. 

Our circuit construction (illustrated for $k=4$ in Figure~\ref{fig:wdb}) implements the following steps:
\begin{align}
	\ket{0^k}&\ket{0^{k-\ell}1^{\ell}}															\notag 		\\
	& \mapsto \ket{0^k}\ket{0^{k-\ell}1^1 0^{\ell-1}}													\label{step1}	\\
	& \mapsto \frac{1}{\surd\tbinom{n}{\ell}} \sum_{i=0}^{\ell} \sqrt{\tbinom{m}{i}\tbinom{n-m}{\ell-i}} \ket{0^{k-i}1^i} \ket{0^{k-\ell}1^1 0^{\ell-1}}	\label{step2}	\\
	& \mapsto \frac{1}{\surd\tbinom{n}{\ell}} \sum_{i=0}^{\ell} \sqrt{\tbinom{m}{i}\tbinom{n-m}{\ell-i}} \ket{0^{k-i}1^i} \ket{0^{k-\ell}1^{\ell}}		\label{step3}	\\
	& \mapsto \frac{1}{\surd\tbinom{n}{\ell}} \sum_{i=0}^{\ell} \sqrt{\tbinom{m}{i}\tbinom{n-m}{\ell-i}} \ket{0^{k-i}1^i} \ket{0^{k+i-\ell}1^{\ell-i}}	\label{step4}	
\end{align}

In Step~\eqref{step1}, we transform the unary encoding of $\ell$ into a one-hot-encoding of $\ell$.
This allows for simpler one-qubit-controlled state additions into the $\ket{0^k}$-initialized second register during Step~\eqref{step2}. This step is illustrated in detail in Figure~\ref{fig:add-subtract}~(left). 
In Step~\eqref{step3}, we revert the one-hot-encoding of the first register into a unary encoding.
Finally, in Step~\eqref{step4}, we subtract (in superposition) from the first register the value of the second register.
This unary subtraction circuit is illustrated in detail in Figure~\ref{fig:add-subtract}~(right) 
and given as a Quirk circuit 
\href{https://algassert.com/quirk#circuit=%7B%22cols%22%3A%5B%5B%22X%22%2C%22X%22%2C%22X%22%2C1%2C%22X%22%2C%22X%22%5D%2C%5B%22Chance%22%2C%22Chance%22%2C%22Chance%22%2C%22Chance%22%2C%22Chance%22%2C%22Chance%22%2C%22Chance%22%2C%22Chance%22%5D%2C%5B1%2C1%2C1%2C%22X%22%2C1%2C1%2C1%2C%22%E2%80%A2%22%5D%2C%5B1%2C1%2C%22%E2%80%A2%22%2C%22X%22%5D%2C%5B1%2C1%2C%22X%22%2C%22%E2%80%A2%22%2C1%2C1%2C%22%E2%80%A2%22%5D%2C%5B1%2C1%2C%22%E2%80%A2%22%2C%22X%22%5D%2C%5B1%2C1%2C1%2C%22X%22%2C1%2C1%2C%22%E2%80%A2%22%5D%2C%5B1%2C%22%E2%80%A2%22%2C%22X%22%5D%2C%5B1%2C%22X%22%2C%22%E2%80%A2%22%2C1%2C1%2C%22%E2%80%A2%22%5D%2C%5B1%2C%22%E2%80%A2%22%2C%22X%22%5D%2C%5B1%2C1%2C%22%E2%80%A2%22%2C%22X%22%5D%2C%5B1%2C1%2C%22X%22%2C%22%E2%80%A2%22%2C1%2C%22%E2%80%A2%22%5D%2C%5B1%2C1%2C%22%E2%80%A2%22%2C%22X%22%5D%2C%5B1%2C1%2C1%2C%22X%22%2C1%2C%22%E2%80%A2%22%5D%2C%5B%22%E2%80%A2%22%2C%22X%22%5D%2C%5B%22X%22%2C%22%E2%80%A2%22%2C1%2C1%2C%22%E2%80%A2%22%5D%2C%5B%22%E2%80%A2%22%2C%22X%22%5D%2C%5B1%2C%22%E2%80%A2%22%2C%22X%22%5D%2C%5B1%2C%22X%22%2C%22%E2%80%A2%22%2C1%2C%22%E2%80%A2%22%5D%2C%5B1%2C%22%E2%80%A2%22%2C%22X%22%5D%2C%5B1%2C1%2C%22%E2%80%A2%22%2C%22X%22%5D%2C%5B1%2C1%2C%22X%22%2C%22%E2%80%A2%22%2C%22%E2%80%A2%22%5D%2C%5B1%2C1%2C%22%E2%80%A2%22%2C%22X%22%5D%2C%5B1%2C1%2C1%2C%22X%22%2C%22%E2%80%A2%22%5D%5D%7D
}{here}.

%% Weight Distribution Block Details
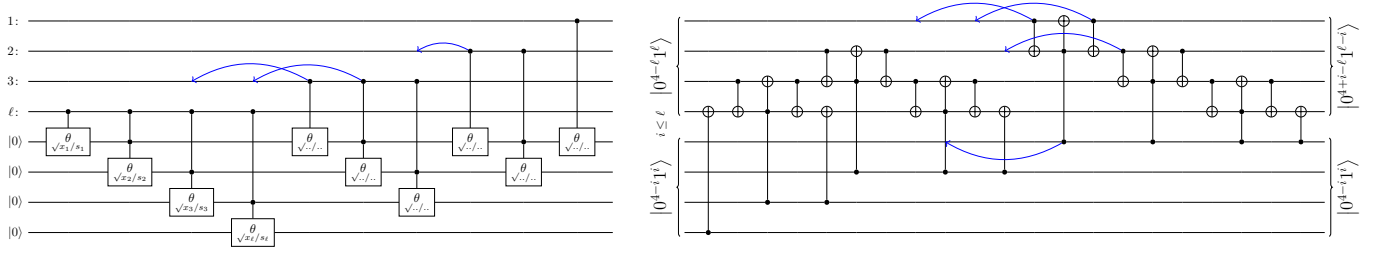
\begin{figure*}[t!]
	\centering
	\begin{adjustbox}{width=\linewidth}
	%% Grover Mixer QAOA
	%% -----------------
	%\hspace*{-0.5cm}
		\begin{quantikz}[row sep={24pt,between origins},execute at end picture={
				\draw[blue,semithick,->]	(\tikzcdmatrixname-3-6) edge[bend right] (\tikzcdmatrixname-3-4);
				\draw[blue,semithick,->]	(\tikzcdmatrixname-3-7) edge[bend right] (\tikzcdmatrixname-3-5);
				\draw[blue,semithick,->]	(\tikzcdmatrixname-2-9) edge[bend right] (\tikzcdmatrixname-2-8);
				\draw[blue,semithick,->]	(\tikzcdmatrixname-1-28) edge[bend right] (\tikzcdmatrixname-1-24);
				\draw[blue,semithick,->]	(\tikzcdmatrixname-5-29) edge[bend left] (\tikzcdmatrixname-5-25);
				\draw[blue,semithick,->]	(\tikzcdmatrixname-1-30) edge[bend right] (\tikzcdmatrixname-1-26);
				\draw[blue,semithick,->]	(\tikzcdmatrixname-2-31) edge[bend right] (\tikzcdmatrixname-2-27);
			}]
			\lstick{$1\colon$}	& \qw			& \qw			& \qw			& \qw				& \qw			& \qw			& \qw			& \qw			& \qw			& \ctrl{4}		& \qw	& 	& 	& 	& 	\lstick[4]{\rotatebox{90}{\Large$\ket{0^{4-\ell}1^{\ell}}$}}	& \qw		& \qw		& \qw		& \qw		& \qw		& \qw		& \qw		& \qw		& \qw		& \qw		& \qw		& \ctrl{1}	& \targ{}	& \ctrl{1}	& \qw		& \qw		& \qw		& \qw		& \qw		& \qw		& \qw		& \qw\rstick[4]{\rotatebox{90}{\Large$\ket{0^{4+i-\ell}1^{\ell-i}}$}}			\\
			\lstick{$2\colon$}	& \qw			& \qw			& \qw			& \qw				& \qw			& \qw			& \qw			& \ctrl{3}		& \ctrl{3}		& \qw			& \qw	& 	& 	& 	& 									& \qw		& \qw		& \qw		& \qw		& \ctrl{1}	& \targ{}	& \ctrl{1}	& \qw		& \qw		& \qw		& \qw		& \targ{}	& \ctrl{-1}	& \targ{}	& \ctrl{1}	& \targ{}	& \ctrl{1}	& \qw		& \qw		& \qw		& \qw		& \qw		\\
			\lstick{$3\colon$}	& \qw			& \qw			& \qw			& \qw				& \ctrl{2}		& \ctrl{2}		& \ctrl{3}		& \qw			& \qw			& \qw			& \qw	& 	& 	& 	& 									& \qw		& \ctrl{1}	& \targ{}	& \ctrl{1}	& \targ{}	& \ctrl{-1}	& \targ{}	& \ctrl{1}	& \targ{}	& \ctrl{1}	& \qw		& \qw		& \qw		& \qw		& \targ{}	& \ctrl{-1}	& \targ{}	& \ctrl{1}	& \targ{}	& \ctrl{1}	& \qw		& \qw		\\
			\lstick{$\ell\colon$}	& \ctrl{1}		& \ctrl{1}		& \ctrl{2}		& \ctrl{3}			& \qw			& \qw			& \qw			& \qw			& \qw			& \qw			& \qw	& 	& 	& 	& 	\lstick{\rotatebox{90}{$i\leq \ell$\qquad}\quad}		& \targ{}	& \targ{}	& \ctrl{-1}	& \targ{}	& \targ{}	& \qw		& \qw		& \targ{}	& \ctrl{-1}	& \targ{}	& \targ{}	& \qw		& \qw		& \qw		& \qw		& \qw		& \qw		& \targ{}	& \ctrl{-1}	& \targ{}	& \targ{}	& \qw		\\
			\lstick{\ket{0}}	& \rygate{x_1}{s_1}	& \ctrl{1}	   	& \qw			& \qw				& \rygate{..}{..}	& \ctrl{1}		& \qw			& \rygate{..}{..}	& \ctrl{1}		& \rygate{..}{..}	& \qw	& 	& 	& 	& 	\lstick[4]{\rotatebox{90}{\Large$\ket{0^{4-i}1^{i}}$}}		& \qw		& \qw		& \qw		& \qw		& \qw		& \qw		& \qw		& \qw		& \qw		& \qw		& \qw		& \qw		& \ctrl{-3}	& \qw		& \qw		& \ctrl{-2}	& \qw		& \qw		& \ctrl{-1}	& \qw		& \ctrl{-1}	& \qw\rstick[4]{\rotatebox{90}{\Large$\ket{0^{4-i}1^{i}}$}}			\\
			\lstick{\ket{0}}	& \qw			& \rygate{x_2}{s_2}	& \ctrl{1}	   	& \qw				& \qw			& \rygate{..}{..}	& \ctrl{1}	   	& \qw			& \rygate{..}{..}	& \qw			& \qw	& 	& 	& 	& 									& \qw		& \qw		& \qw		& \qw		& \qw		& \ctrl{-3}	& \qw		& \qw		& \ctrl{-2}	& \qw		& \ctrl{-2}	& \qw		& \qw		& \qw		& \qw		& \qw		& \qw		& \qw		& \qw		& \qw		& \qw		& \qw		\\
			\lstick{\ket{0}}	& \qw			& \qw			& \rygate{x_3}{s_3}	& \ctrl{1}	   		& \qw			& \qw			& \rygate{..}{..}	& \qw			& \qw			& \qw			& \qw	& 	& 	& 	& 									& \qw		& \qw		& \ctrl{-3}	& \qw		& \ctrl{-3}	& \qw		& \qw		& \qw		& \qw		& \qw		& \qw		& \qw		& \qw		& \qw		& \qw		& \qw		& \qw		& \qw		& \qw		& \qw		& \qw		& \qw		\\	
			\lstick{\ket{0}}	& \qw			& \qw			& \qw			& \rygate{x_{\ell}}{s_{\ell}}	& \qw			& \qw			& \qw			& \qw			& \qw			& \qw			& \qw	& 	& 	& 	& 									& \ctrl{-4}	& \qw		& \qw		& \qw		& \qw		& \qw		& \qw		& \qw		& \qw		& \qw		& \qw		& \qw		& \qw		& \qw		& \qw		& \qw		& \qw		& \qw		& \qw		& \qw		& \qw		& \qw							
		\end{quantikz}
	\end{adjustbox}
	\caption{Detailed implementations of (LEFT) the controlled addition circuit and (RIGHT) the unary subtraction circuit 
		which are used in the weight distribution block $\wdb{n}{m}{k=4}$ in Figure~\ref{fig:wdb}. 
		In both circuits, the gates take on a stair shape on the second register, which enables parallelization of the controlled gates (blue arrows).\newline
		For controlled addition of the state $\smash{\tbinom{n}{\ell}^{-1/2} \sum_{i=0}^{\ell} \surd{\tbinom{m}{i}}\surd{\tbinom{n-m}{\ell-i}} \ket{0^{4-i}1^i}}$ into the second register, 
		we first precompute the terms $x_i = \tbinom{m}{i}\tbinom{n-m}{\ell-i}$ as well as the suffix sums $s_i = \sum_{i=0}^{\ell} x_i$. The state can then be added 
		by a series of controlled $\theta_{\surd x_i/s_i}$-gates.
		To simplify the large controlled permutation gates in the unary subtraction circuit, we decompose controlled permutation gates into a stair of controlled SWAP gates, known as Fredkin gates, 
		followed by the standard decomposition of a Fredkin gate into a Toffoli and two CNOT gates.\newline 
		For LNN topologies, we combine the circuit with a SWAP network that keeps the controls of one register in direct neighborhood to the gates in the second register.
		This still allows for parallelization, at an overhead of additional $\bigO(k)$ depth / $\bigO(k^2)$ gates.
	}
	\label{fig:add-subtract}
\end{figure*}

We now derive costs from the presented circuits in Figures~\ref{fig:wdb}~\&~\ref{fig:add-subtract}:
Clearly, the encoding changes in Steps~\eqref{step1}~\&~\eqref{step3} need only $\bigO(k)$ gates and depth. 
For the controlled addition in Step~\eqref{step2}, each possible value $0\leq \ell \leq k$ gives rise to a controlled addition that needs $\bigO(\ell)$ gates and depth.
Since the addition gates have a stair shape, they can be partially parallelized, and hence we get a depth of $\bigO(k)$ for a total gate count of $\bigO(k^2)$.
The same holds for the controlled permutation gates which form a building block of the unary subtraction in Step~\eqref{step4}.
Hence the overall depth of our circuit for $\wdb{n}{m}{k}$ is $\bigO(k)$ and the total gate count is $\bigO(k^2)$.

Additional details need to be discussed here:
First, if $m<k$, then the second register needs to be truncated from $k$ to $m$ qubits. This can be achieved by removing any operations involving the most-significant $k-m$ qubits of the second register.
This is unproblematic, as these operations correspond to terms $\tbinom{m}{i}=0$ for $i>m$ in Step~\eqref{step3}, and thus to identities or $0$-valued $1$-controls in the Step~\eqref{step4}.
We illustrate this with interactive Quirk circuits for weight distribution blocks
\href{https://algassert.com/quirk#circuit=%7B%22cols%22%3A%5B%5B%22~jat7%22%5D%2C%5B%22Chance%22%2C%22Chance%22%2C%22Chance%22%5D%2C%5B%22X%22%2C%22%E2%80%A2%22%5D%2C%5B1%2C%22X%22%2C%22%E2%80%A2%22%5D%2C%5B%22Chance%22%2C%22Chance%22%2C%22Chance%22%2C%22Chance3%22%5D%2C%5B1%2C1%2C%22%E2%80%A2%22%2C%22~2bbd%22%5D%2C%5B1%2C1%2C%22%E2%80%A2%22%2C%22%E2%80%A2%22%2C%22~e5cj%22%5D%2C%5B1%2C1%2C%22%E2%80%A2%22%2C1%2C%22%E2%80%A2%22%2C%22~i81f%22%5D%2C%5B1%2C%22%E2%80%A2%22%2C1%2C%22~ppan%22%5D%2C%5B1%2C%22%E2%80%A2%22%2C1%2C%22%E2%80%A2%22%2C%22~2mea%22%5D%2C%5B%22%E2%80%A2%22%2C1%2C1%2C%22~gv9n%22%5D%2C%5B%22Chance%22%2C%22Chance%22%2C%22Chance%22%2C%22Chance3%22%5D%2C%5B1%2C%22X%22%2C%22%E2%80%A2%22%5D%2C%5B%22X%22%2C%22%E2%80%A2%22%5D%2C%5B%22Chance%22%2C%22Chance%22%2C%22Chance%22%2C%22Chance3%22%5D%2C%5B1%2C1%2C%22X%22%2C1%2C1%2C%22%E2%80%A2%22%5D%2C%5B1%2C%22%3E%3E2%22%2C1%2C1%2C%22%E2%80%A2%22%5D%2C%5B1%2C1%2C%22X%22%2C1%2C%22%E2%80%A2%22%5D%2C%5B%22%3E%3E3%22%2C1%2C1%2C%22%E2%80%A2%22%5D%2C%5B1%2C1%2C%22X%22%2C%22%E2%80%A2%22%5D%2C%5B%22Chance3%22%2C1%2C1%2C%22Chance3%22%5D%2C%5B%22~30he%22%2C1%2C1%2C%22~30he%22%5D%5D%2C%22gates%22%3A%5B%7B%22id%22%3A%22~6dbg%22%2C%22name%22%3A%22%CE%B8%5Cn%E2%88%9A1%2F2%22%2C%22circuit%22%3A%7B%22cols%22%3A%5B%5B%7B%22id%22%3A%22Ryft%22%2C%22arg%22%3A%222acos(sqrt(1%2F2))%22%7D%5D%5D%7D%7D%2C%7B%22id%22%3A%22~16b9%22%2C%22name%22%3A%22%CE%B8%5Cn%E2%88%9A1%2F3%22%2C%22circuit%22%3A%7B%22cols%22%3A%5B%5B%7B%22id%22%3A%22Ryft%22%2C%22arg%22%3A%222acos(sqrt(1%2F3))%22%7D%5D%5D%7D%7D%2C%7B%22id%22%3A%22~qdji%22%2C%22name%22%3A%22%CE%B8%5Cn%E2%88%9A2%2F3%22%2C%22circuit%22%3A%7B%22cols%22%3A%5B%5B%7B%22id%22%3A%22Ryft%22%2C%22arg%22%3A%222acos(sqrt(2%2F3))%22%7D%5D%5D%7D%7D%2C%7B%22id%22%3A%22~jat7%22%2C%22name%22%3A%22Test%22%2C%22circuit%22%3A%7B%22cols%22%3A%5B%5B%22Counting2%22%5D%2C%5B%22%E2%80%A2%22%2C%22%E2%80%A2%22%2C%22X%22%5D%2C%5B%22X%22%2C%22%E2%80%A2%22%5D%2C%5B%22X%22%2C1%2C%22%E2%80%A2%22%5D%5D%7D%7D%2C%7B%22id%22%3A%22~j9or%22%2C%22name%22%3A%22U2%2C2%22%2C%22circuit%22%3A%7B%22cols%22%3A%5B%5B%22X%22%2C%22%E2%80%A2%22%5D%2C%5B%22%E2%80%A2%22%2C%22~6dbg%22%5D%2C%5B%22X%22%2C%22%E2%80%A2%22%5D%5D%7D%7D%2C%7B%22id%22%3A%22~30he%22%2C%22name%22%3A%22U3%2C3%22%2C%22circuit%22%3A%7B%22cols%22%3A%5B%5B%22X%22%2C%22%E2%80%A2%22%5D%2C%5B%22%E2%80%A2%22%2C%22~16b9%22%5D%2C%5B%22X%22%2C%22%E2%80%A2%22%5D%2C%5B%22X%22%2C1%2C%22%E2%80%A2%22%5D%2C%5B%22%E2%80%A2%22%2C%22%E2%80%A2%22%2C%22~qdji%22%5D%2C%5B%22X%22%2C1%2C%22%E2%80%A2%22%5D%2C%5B1%2C%22X%22%2C%22%E2%80%A2%22%5D%2C%5B1%2C%22%E2%80%A2%22%2C%22~6dbg%22%5D%2C%5B1%2C%22X%22%2C%22%E2%80%A2%22%5D%5D%7D%7D%2C%7B%22id%22%3A%22~3t19%22%2C%22name%22%3A%22%CE%B8%5Cn%E2%88%9A20%2F165%22%2C%22circuit%22%3A%7B%22cols%22%3A%5B%5B%7B%22id%22%3A%22Ryft%22%2C%22arg%22%3A%222acos(sqrt(20%2F165))%22%7D%5D%5D%7D%7D%2C%7B%22id%22%3A%22~dh7n%22%2C%22name%22%3A%22%CE%B8%5Cn%E2%88%9A75%2F145%22%2C%22circuit%22%3A%7B%22cols%22%3A%5B%5B%7B%22id%22%3A%22Ryft%22%2C%22arg%22%3A%222acos(sqrt(75%2F145))%22%7D%5D%5D%7D%7D%2C%7B%22id%22%3A%22~cq6s%22%2C%22name%22%3A%22%CE%B8%5Cn%E2%88%9A60%2F70%22%2C%22circuit%22%3A%7B%22cols%22%3A%5B%5B%7B%22id%22%3A%22Ryft%22%2C%22arg%22%3A%222acos(sqrt(60%2F70))%22%7D%5D%5D%7D%7D%2C%7B%22id%22%3A%22~m5i%22%2C%22name%22%3A%22%CE%B8%5Cn%E2%88%9A15%2F55%22%2C%22circuit%22%3A%7B%22cols%22%3A%5B%5B%7B%22id%22%3A%22Ryft%22%2C%22arg%22%3A%222acos(sqrt(15%2F55))%22%7D%5D%5D%7D%7D%2C%7B%22id%22%3A%22~i8ql%22%2C%22name%22%3A%22%CE%B8%5Cn%E2%88%9A30%2F40%22%2C%22circuit%22%3A%7B%22cols%22%3A%5B%5B%7B%22id%22%3A%22Ryft%22%2C%22arg%22%3A%222acos(sqrt(30%2F40))%22%7D%5D%5D%7D%7D%2C%7B%22id%22%3A%22~kcer%22%2C%22name%22%3A%22%CE%B8%5Cn%E2%88%9A6%2F11%22%2C%22circuit%22%3A%7B%22cols%22%3A%5B%5B%7B%22id%22%3A%22Ryft%22%2C%22arg%22%3A%222acos(sqrt(6%2F11))%22%7D%5D%5D%7D%7D%2C%7B%22id%22%3A%22~2bbd%22%2C%22name%22%3A%22%CE%B8%5Cn%E2%88%9A1%2F20%22%2C%22circuit%22%3A%7B%22cols%22%3A%5B%5B%7B%22id%22%3A%22Ryft%22%2C%22arg%22%3A%222acos(sqrt(1%2F20))%22%7D%5D%5D%7D%7D%2C%7B%22id%22%3A%22~e5cj%22%2C%22name%22%3A%22%CE%B8%5Cn%E2%88%9A9%2F19%22%2C%22circuit%22%3A%7B%22cols%22%3A%5B%5B%7B%22id%22%3A%22Ryft%22%2C%22arg%22%3A%222acos(sqrt(9%2F19))%22%7D%5D%5D%7D%7D%2C%7B%22id%22%3A%22~i81f%22%2C%22name%22%3A%22%CE%B8%5Cn%E2%88%9A9%2F10%22%2C%22circuit%22%3A%7B%22cols%22%3A%5B%5B%7B%22id%22%3A%22Ryft%22%2C%22arg%22%3A%222acos(sqrt(9%2F10))%22%7D%5D%5D%7D%7D%2C%7B%22id%22%3A%22~ppan%22%2C%22name%22%3A%22%CE%B8%5Cn%E2%88%9A3%2F15%22%2C%22circuit%22%3A%7B%22cols%22%3A%5B%5B%7B%22id%22%3A%22Ryft%22%2C%22arg%22%3A%222acos(sqrt(3%2F15))%22%7D%5D%5D%7D%7D%2C%7B%22id%22%3A%22~2mea%22%2C%22name%22%3A%22%CE%B8%5Cn%E2%88%9A9%2F12%22%2C%22circuit%22%3A%7B%22cols%22%3A%5B%5B%7B%22id%22%3A%22Ryft%22%2C%22arg%22%3A%222acos(sqrt(9%2F12))%22%7D%5D%5D%7D%7D%2C%7B%22id%22%3A%22~gv9n%22%2C%22name%22%3A%22%CE%B8%5Cn%E2%88%9A3%2F6%22%2C%22circuit%22%3A%7B%22cols%22%3A%5B%5B%7B%22id%22%3A%22Ryft%22%2C%22arg%22%3A%222acos(sqrt(3%2F6))%22%7D%5D%5D%7D%7D%2C%7B%22id%22%3A%22~o2im%22%2C%22name%22%3A%22%CE%B8%5Cn%E2%88%9A1%2F10%22%2C%22circuit%22%3A%7B%22cols%22%3A%5B%5B%7B%22id%22%3A%22Ryft%22%2C%22arg%22%3A%222acos(sqrt(1%2F10))%22%7D%5D%5D%7D%7D%2C%7B%22id%22%3A%22~9u45%22%2C%22name%22%3A%22%CE%B8%5Cn%E2%88%9A6%2F9%22%2C%22circuit%22%3A%7B%22cols%22%3A%5B%5B%7B%22id%22%3A%22Ryft%22%2C%22arg%22%3A%222acos(sqrt(6%2F9))%22%7D%5D%5D%7D%7D%2C%7B%22id%22%3A%22~v2k1%22%2C%22name%22%3A%22%CE%B8%5Cn%E2%88%9A3%2F10%22%2C%22circuit%22%3A%7B%22cols%22%3A%5B%5B%7B%22id%22%3A%22Ryft%22%2C%22arg%22%3A%222acos(sqrt(3%2F10))%22%7D%5D%5D%7D%7D%2C%7B%22id%22%3A%22~4npg%22%2C%22name%22%3A%22%CE%B8%5Cn%E2%88%9A6%2F7%22%2C%22circuit%22%3A%7B%22cols%22%3A%5B%5B%7B%22id%22%3A%22Ryft%22%2C%22arg%22%3A%222acos(sqrt(6%2F7))%22%7D%5D%5D%7D%7D%2C%7B%22id%22%3A%22~l3k5%22%2C%22name%22%3A%22%CE%B8%5Cn%E2%88%9A3%2F5%22%2C%22circuit%22%3A%7B%22cols%22%3A%5B%5B%7B%22id%22%3A%22Ryft%22%2C%22arg%22%3A%222acos(sqrt(3%2F5))%22%7D%5D%5D%7D%7D%5D%7D
}{$\wdb{6}{3}{3}$} and
\href{https://algassert.com/quirk#circuit=%7B%22cols%22%3A%5B%5B%22~jat7%22%5D%2C%5B%22Chance%22%2C%22Chance%22%2C%22Chance%22%5D%2C%5B%22X%22%2C%22%E2%80%A2%22%5D%2C%5B1%2C%22X%22%2C%22%E2%80%A2%22%5D%2C%5B%22Chance%22%2C%22Chance%22%2C%22Chance%22%2C%22Chance2%22%5D%2C%5B1%2C1%2C%22%E2%80%A2%22%2C%22~o2im%22%5D%2C%5B1%2C1%2C%22%E2%80%A2%22%2C%22%E2%80%A2%22%2C%22~9u45%22%5D%2C%5B1%2C%22%E2%80%A2%22%2C1%2C%22~v2k1%22%5D%2C%5B1%2C%22%E2%80%A2%22%2C1%2C%22%E2%80%A2%22%2C%22~4npg%22%5D%2C%5B%22%E2%80%A2%22%2C1%2C1%2C%22~l3k5%22%5D%2C%5B%22Chance%22%2C%22Chance%22%2C%22Chance%22%2C%22Chance2%22%5D%2C%5B1%2C%22X%22%2C%22%E2%80%A2%22%5D%2C%5B%22X%22%2C%22%E2%80%A2%22%5D%2C%5B%22Chance%22%2C%22Chance%22%2C%22Chance%22%2C%22Chance2%22%5D%2C%5B1%2C%22%3E%3E2%22%2C1%2C1%2C%22%E2%80%A2%22%5D%2C%5B1%2C1%2C%22X%22%2C1%2C%22%E2%80%A2%22%5D%2C%5B%22%3E%3E3%22%2C1%2C1%2C%22%E2%80%A2%22%5D%2C%5B1%2C1%2C%22X%22%2C%22%E2%80%A2%22%5D%2C%5B%22Chance3%22%2C1%2C1%2C%22Chance2%22%5D%2C%5B%22~30he%22%2C1%2C1%2C%22~j9or%22%5D%5D%2C%22gates%22%3A%5B%7B%22id%22%3A%22~6dbg%22%2C%22name%22%3A%22%CE%B8%5Cn%E2%88%9A1%2F2%22%2C%22circuit%22%3A%7B%22cols%22%3A%5B%5B%7B%22id%22%3A%22Ryft%22%2C%22arg%22%3A%222acos(sqrt(1%2F2))%22%7D%5D%5D%7D%7D%2C%7B%22id%22%3A%22~16b9%22%2C%22name%22%3A%22%CE%B8%5Cn%E2%88%9A1%2F3%22%2C%22circuit%22%3A%7B%22cols%22%3A%5B%5B%7B%22id%22%3A%22Ryft%22%2C%22arg%22%3A%222acos(sqrt(1%2F3))%22%7D%5D%5D%7D%7D%2C%7B%22id%22%3A%22~qdji%22%2C%22name%22%3A%22%CE%B8%5Cn%E2%88%9A2%2F3%22%2C%22circuit%22%3A%7B%22cols%22%3A%5B%5B%7B%22id%22%3A%22Ryft%22%2C%22arg%22%3A%222acos(sqrt(2%2F3))%22%7D%5D%5D%7D%7D%2C%7B%22id%22%3A%22~jat7%22%2C%22name%22%3A%22Test%22%2C%22circuit%22%3A%7B%22cols%22%3A%5B%5B%22Counting2%22%5D%2C%5B%22%E2%80%A2%22%2C%22%E2%80%A2%22%2C%22X%22%5D%2C%5B%22X%22%2C%22%E2%80%A2%22%5D%2C%5B%22X%22%2C1%2C%22%E2%80%A2%22%5D%5D%7D%7D%2C%7B%22id%22%3A%22~j9or%22%2C%22name%22%3A%22U2%2C2%22%2C%22circuit%22%3A%7B%22cols%22%3A%5B%5B%22X%22%2C%22%E2%80%A2%22%5D%2C%5B%22%E2%80%A2%22%2C%22~6dbg%22%5D%2C%5B%22X%22%2C%22%E2%80%A2%22%5D%5D%7D%7D%2C%7B%22id%22%3A%22~30he%22%2C%22name%22%3A%22U3%2C3%22%2C%22circuit%22%3A%7B%22cols%22%3A%5B%5B%22X%22%2C%22%E2%80%A2%22%5D%2C%5B%22%E2%80%A2%22%2C%22~16b9%22%5D%2C%5B%22X%22%2C%22%E2%80%A2%22%5D%2C%5B%22X%22%2C1%2C%22%E2%80%A2%22%5D%2C%5B%22%E2%80%A2%22%2C%22%E2%80%A2%22%2C%22~qdji%22%5D%2C%5B%22X%22%2C1%2C%22%E2%80%A2%22%5D%2C%5B1%2C%22X%22%2C%22%E2%80%A2%22%5D%2C%5B1%2C%22%E2%80%A2%22%2C%22~6dbg%22%5D%2C%5B1%2C%22X%22%2C%22%E2%80%A2%22%5D%5D%7D%7D%2C%7B%22id%22%3A%22~3t19%22%2C%22name%22%3A%22%CE%B8%5Cn%E2%88%9A20%2F165%22%2C%22circuit%22%3A%7B%22cols%22%3A%5B%5B%7B%22id%22%3A%22Ryft%22%2C%22arg%22%3A%222acos(sqrt(20%2F165))%22%7D%5D%5D%7D%7D%2C%7B%22id%22%3A%22~dh7n%22%2C%22name%22%3A%22%CE%B8%5Cn%E2%88%9A75%2F145%22%2C%22circuit%22%3A%7B%22cols%22%3A%5B%5B%7B%22id%22%3A%22Ryft%22%2C%22arg%22%3A%222acos(sqrt(75%2F145))%22%7D%5D%5D%7D%7D%2C%7B%22id%22%3A%22~cq6s%22%2C%22name%22%3A%22%CE%B8%5Cn%E2%88%9A60%2F70%22%2C%22circuit%22%3A%7B%22cols%22%3A%5B%5B%7B%22id%22%3A%22Ryft%22%2C%22arg%22%3A%222acos(sqrt(60%2F70))%22%7D%5D%5D%7D%7D%2C%7B%22id%22%3A%22~m5i%22%2C%22name%22%3A%22%CE%B8%5Cn%E2%88%9A15%2F55%22%2C%22circuit%22%3A%7B%22cols%22%3A%5B%5B%7B%22id%22%3A%22Ryft%22%2C%22arg%22%3A%222acos(sqrt(15%2F55))%22%7D%5D%5D%7D%7D%2C%7B%22id%22%3A%22~i8ql%22%2C%22name%22%3A%22%CE%B8%5Cn%E2%88%9A30%2F40%22%2C%22circuit%22%3A%7B%22cols%22%3A%5B%5B%7B%22id%22%3A%22Ryft%22%2C%22arg%22%3A%222acos(sqrt(30%2F40))%22%7D%5D%5D%7D%7D%2C%7B%22id%22%3A%22~kcer%22%2C%22name%22%3A%22%CE%B8%5Cn%E2%88%9A6%2F11%22%2C%22circuit%22%3A%7B%22cols%22%3A%5B%5B%7B%22id%22%3A%22Ryft%22%2C%22arg%22%3A%222acos(sqrt(6%2F11))%22%7D%5D%5D%7D%7D%2C%7B%22id%22%3A%22~2bbd%22%2C%22name%22%3A%22%CE%B8%5Cn%E2%88%9A1%2F20%22%2C%22circuit%22%3A%7B%22cols%22%3A%5B%5B%7B%22id%22%3A%22Ryft%22%2C%22arg%22%3A%222acos(sqrt(1%2F20))%22%7D%5D%5D%7D%7D%2C%7B%22id%22%3A%22~e5cj%22%2C%22name%22%3A%22%CE%B8%5Cn%E2%88%9A9%2F19%22%2C%22circuit%22%3A%7B%22cols%22%3A%5B%5B%7B%22id%22%3A%22Ryft%22%2C%22arg%22%3A%222acos(sqrt(9%2F19))%22%7D%5D%5D%7D%7D%2C%7B%22id%22%3A%22~i81f%22%2C%22name%22%3A%22%CE%B8%5Cn%E2%88%9A9%2F10%22%2C%22circuit%22%3A%7B%22cols%22%3A%5B%5B%7B%22id%22%3A%22Ryft%22%2C%22arg%22%3A%222acos(sqrt(9%2F10))%22%7D%5D%5D%7D%7D%2C%7B%22id%22%3A%22~ppan%22%2C%22name%22%3A%22%CE%B8%5Cn%E2%88%9A3%2F15%22%2C%22circuit%22%3A%7B%22cols%22%3A%5B%5B%7B%22id%22%3A%22Ryft%22%2C%22arg%22%3A%222acos(sqrt(3%2F15))%22%7D%5D%5D%7D%7D%2C%7B%22id%22%3A%22~2mea%22%2C%22name%22%3A%22%CE%B8%5Cn%E2%88%9A9%2F12%22%2C%22circuit%22%3A%7B%22cols%22%3A%5B%5B%7B%22id%22%3A%22Ryft%22%2C%22arg%22%3A%222acos(sqrt(9%2F12))%22%7D%5D%5D%7D%7D%2C%7B%22id%22%3A%22~gv9n%22%2C%22name%22%3A%22%CE%B8%5Cn%E2%88%9A3%2F6%22%2C%22circuit%22%3A%7B%22cols%22%3A%5B%5B%7B%22id%22%3A%22Ryft%22%2C%22arg%22%3A%222acos(sqrt(3%2F6))%22%7D%5D%5D%7D%7D%2C%7B%22id%22%3A%22~o2im%22%2C%22name%22%3A%22%CE%B8%5Cn%E2%88%9A1%2F10%22%2C%22circuit%22%3A%7B%22cols%22%3A%5B%5B%7B%22id%22%3A%22Ryft%22%2C%22arg%22%3A%222acos(sqrt(1%2F10))%22%7D%5D%5D%7D%7D%2C%7B%22id%22%3A%22~9u45%22%2C%22name%22%3A%22%CE%B8%5Cn%E2%88%9A6%2F9%22%2C%22circuit%22%3A%7B%22cols%22%3A%5B%5B%7B%22id%22%3A%22Ryft%22%2C%22arg%22%3A%222acos(sqrt(6%2F9))%22%7D%5D%5D%7D%7D%2C%7B%22id%22%3A%22~v2k1%22%2C%22name%22%3A%22%CE%B8%5Cn%E2%88%9A3%2F10%22%2C%22circuit%22%3A%7B%22cols%22%3A%5B%5B%7B%22id%22%3A%22Ryft%22%2C%22arg%22%3A%222acos(sqrt(3%2F10))%22%7D%5D%5D%7D%7D%2C%7B%22id%22%3A%22~4npg%22%2C%22name%22%3A%22%CE%B8%5Cn%E2%88%9A6%2F7%22%2C%22circuit%22%3A%7B%22cols%22%3A%5B%5B%7B%22id%22%3A%22Ryft%22%2C%22arg%22%3A%222acos(sqrt(6%2F7))%22%7D%5D%5D%7D%7D%2C%7B%22id%22%3A%22~l3k5%22%2C%22name%22%3A%22%CE%B8%5Cn%E2%88%9A3%2F5%22%2C%22circuit%22%3A%7B%22cols%22%3A%5B%5B%7B%22id%22%3A%22Ryft%22%2C%22arg%22%3A%222acos(sqrt(3%2F5))%22%7D%5D%5D%7D%7D%5D%7D
}{$\wdb{5}{2}{3}$},
which are both used in the construction of the Dicke state $\dicke{11}{3}$ given in Figure~\ref{fig:our-method}.

Secondly, we would like to have an implementation for $\wdb{n}{m}{k}$ that also works on LNN topologies along the $k+k$ qubits with only a constant factor overhead. 
This is possible: We can spend an additional $\bigO(k)$ depth and $\bigO(k^2)$ gates on a SWAP network which for the circuits in Steps~\eqref{step2}~\&~\eqref{step4}
keeps the controls of one register in direct neighborhood to the controlled gates in the second register. This still allows for parallelization (albeit by a constant factor less).
We get:

\begin{lemma}
	The weight distribution block $\wdb{n}{m}{k}$ can be implemented with a circuit of depth $\bigO(k)$ on LNN connectivities on and between 
	$\min(n-m,k)+\min(m,k)$ qubits, with $\bigO(k^2)$ 2-qubit gates and without ancilla qubits.
	\label{lem:wdb-LNN}
\end{lemma}

\begin{figure*}[t]
	\centering
	\includegraphics[width=\textwidth]{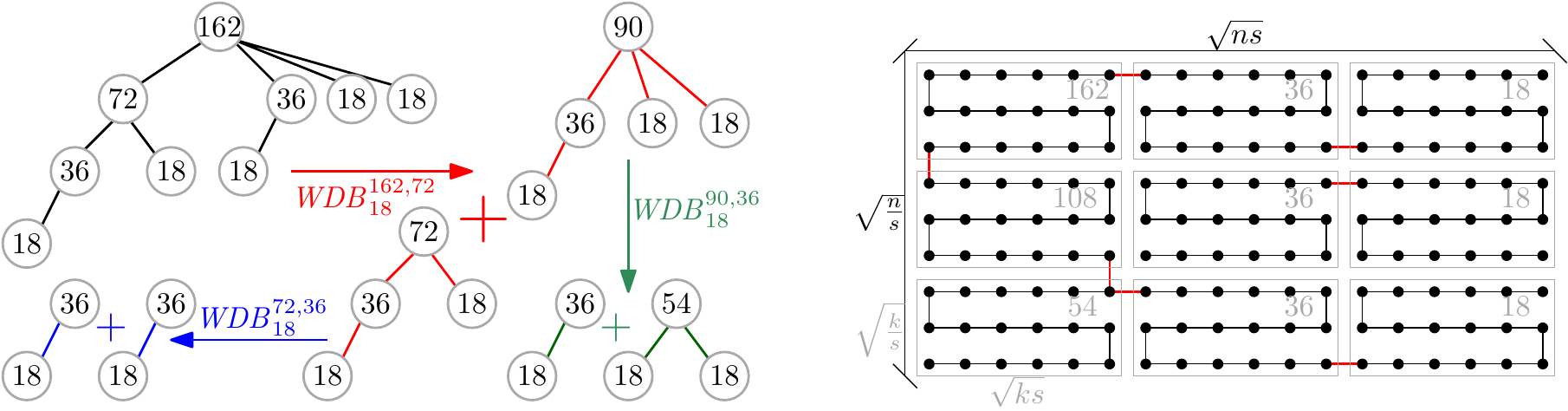}
	\caption{Structures for the recursive application of weight distribution blocks $\wdb{n}{m}{k}$, illustrated for a Dicke state \dicke{162}{18}:
		(LEFT) For all-to-all connectivity, we distribute the $n$ qubits into $\tfrac{n}{k}$ disjoint sets of size $k$.
		These sets are arranged in a rooted tree and labelled by their \emph{tier}, which denotes the total number of all qubits in their subtree.
		We construct the tree using methods from Union-Find data structures, such that recursive applications of weight distribution blocks $\wdb{x}{y}{k}$ 
		between a root of tier $x$ and its highest-tiered child of tier $y$ splits the tree into subtrees with roots of tier $x,x-y$ with roughly equal degree.\newline
		(RIGHT) For grids of size $\surd\tfrac{n}{s} \times \sqrt{ns}$ with $1\leq s \leq k$, we divide the qubits into $\surd\tfrac{n}{k} \times \surd \tfrac{n}{k}$ 
		many rectangles of size $\surd{\tfrac{k}{s}} \times \sqrt{ks}$. We identify LNN connectivities among the $k$ qubits in each rectangle, such that LNN endpoints
		in adjacent rectangles are adjacent grid points.
		The tier of a rectangle corresponds to the total number of qubits in their lower right quadrant (for rectangles in the leftmost column)
		or to the total number of qubits in their row and to their the right (for all other rectangles). 
		We use weight distribution blocks $\wdb{x}{x-\surd nk}{k}$ to distribute an input weight top-to-bottom in the first column,
		followed by blocks $\wdb{y}{y-k}{k}$ for left-to-right distribution in each row, where the rows can be processed in parallel.
	}
	\label{fig:topologies}
\end{figure*}

Remark that in Lemma 2, we did not specify in which order the LNN connectivities of 
the two registers are ``glued'' together. In fact, both options are fine, as we can 
rearrange the qubits in the first register into a reverse order with a SWAP network of
$\bigO(k)$ depths with $\bigO(k^2)$ SWAP gates. 

%% Recursive Structures
\section{Topology-Dependent Recursive Structure}
\label{sec:recursive}

In this last technical section, we design recursive structures for the application of our weight distribution blocks. 
For all-to-all connectivities, we are not constrained in the design: any recursive structure inherits the all-to-all connectivity between the two sets of $k$ qubits. 
For grid topologies, we have to be more careful: the recursive structure should respect the connectivity requirements of Lemma~\ref{lem:wdb-LNN}.
Examples of our constructions are presented in Figure~\ref{fig:topologies}.

% All-to-all
\subsection{All-to-all connectivity}

To obtain a recursive structure on all-to-all connectivities, we partition the $n$ qubits into $\lfloor \tfrac{n}{k} \rfloor$ disjoint sets of size $k$, plus possibly a set 
of $n \pmod k$ remaining qubits. To each set we assign a \emph{tier} which initially denotes the number of qubits in the set.
We then use a union-by-size (i.e., tier) algorithm (without path compression) for disjoint sets data structures, to build a rooted tree that 
(i) has the qubit sets as its nodes and (ii) has nodes labelled by tiers that correspond to the total number of qubits in all sets of the respective node's subtree, 
see Figure~\ref{fig:topologies}~(left):

Until every set is contained in the same single tree, we take the two trees with lowest-tier roots $a,b$, 
attach the lower-tier root $a$ as a child to the other root $b$ and update $b$'s tier with the sum of both tiers.
The resulting tree has only logarithmic height in the number of nodes, $\bigO(\log \tfrac{n}{k})$.%
\footnote{An alternative ``top-down'' approach would be to start with a designated root, and over several rounds $r$ connect to each of the $2^r$ sets in the tree 
a not-yet connected set, and computing the tiers once the tree is finalized.}

We can now use weight distribution blocks to recursively ``undo'' the union-by-size algorithm: For a root of tier $x$ and its highest-tier child of tier $y$, 
we apply $\wdb{x}{y}{k}$ to the root and its child, cut their edge, and get two smaller trees with root tiers $x-y$ and $y$, respectively. The next iteration 
of weight distribution blocks can now be applied to these two roots (and their respective highest-tier children) in parallel. Due to the logarithmic height
of the tree, after at most $\bigO(\log \tfrac{n}{k})$ recursion steps we have distributed the input Hamming weight across all qubit sets.
This leads to our first main result:

\begin{theorem}[All-to-all connectivity]
	The Dicke state unitary $\dsu{n}{k}$ can be implemented with a circuit of depth $\bigO(k \log \tfrac{n}{k})$ on all-to-all connectivities, 
	with $\bigO(k n)$ 2-qubit gates in total and without ancilla qubits.
	\label{thm:dsu-full}
\end{theorem}
\begin{proof}
	We first follow the scheme outlined above to distribute Hamming weights $\ell < k$ across sets of at most $k$ qubits using weight distribution blocks.
	According to Lemma~\ref{lem:wdb-LNN}, such blocks can be implemented in $\bigO(k)$ depth, with $\bigO(k^2)$ gates and no ancilla qubits. 
	Every set of qubits except for the root is exactly once a target of a weight distribution block, hence all blocks together sum up to $\bigO(\tfrac{n}{k} k^2) = \bigO(kn)$ many gates.
	On the other hand, we have at most a logarithmic number of recursion steps, yielding an overall depth of $\bigO(k \log \tfrac{n}{k})$. 

	Finally, we use Dicke state unitaries $\dsu{k}{k}$ (and $\dsu{n\bmod k}{n\bmod k}$ if $k$ does not divide $n$) in parallel, to distribute Hamming weights inside each set.
	According to Lemma~\ref{lem:dsu-LNN}, this adds another gate count of $\bigO(\tfrac{n}{k} k^2) = \bigO(kn)$ gates with depth $\bigO(k)$, bringing the total counts 
	to the values given in the Theorem.
\end{proof}

We note in passing the similarity to the minimum broadcast time problem~\cite{minimumbroadcasttime}, 
where a message from a (to be chosen) sender should be broadcast to all nodes of the graph in minimum time with the restriction that 
nodes can send copies of the message to at most one of their neighbors in any given time step. 

If the graph is part of the input, the minimum broadcast time problem is NP hard.
If, however, the graph is known to be a clique, then an adaptation of the scheme above achieves the best possible broadcast time.
We thus conjecture that for constant $k$, the depth shown in Theorem~\ref{thm:dsu-full} is optimal up to constant factors.

% Grid 
\subsection{Grid connectivity}

We now consider any grid connectivity of grid size $\Omega(\surd \tfrac{n}{s}) \times \bigO( \sqrt{ns})$ with $1\leq s \leq k$. 
To start with, we assume that the values $\surd\tfrac{n}{s}, \surd\tfrac{k}{s}, \surd\tfrac{n}{k}$ as well as $\sqrt{ns}, \sqrt{ks}, \sqrt{nk}$ are all integers.
In this case, we can divide a rectangular grid of size $\surd \tfrac{n}{s} \times  \sqrt{n}s$ into $\surd\tfrac{n}{k} \times \surd \tfrac{n}{k}$ many rectangles of size
$\surd{\tfrac{k}{s}} \times \sqrt{ks}$, containing $k$ qubits each. In each rectangle, we single out a snake-like LNN subtopology such that the snake's ``head'' and ``tail''
are adjacent to either the head or the tail of a snake in each neighboring rectangle, see Figure~\ref{fig:topologies}~(right). 

Now we can again assign tiers to the rectangles, similar to the previous subsection:
In each column except for the leftmost, the tier of a rectangle corresponds to the total number of qubits in that rectangle and all rectangles to its right.
In the leftmost column, the tier of a rectangle corresponds to the total number of qubits in that rectangle and all rectangles in its lower right quadrant.

We use weight distribution blocks $\wdb{x}{x-\sqrt{nk}}{k}$ to distribute an input Hamming weight from the top left rectangle top-to-bottom in the leftmost column,
followed by weight distribution blocks $\wdb{y}{y-k}{k}$ for left-to-right distribution in each row, where the rows can be processed in parallel.
This leads to our second main result:

\begin{theorem}[Grid connectivity]
	The Dicke state unitary $\dsu{n}{k}$ can be implemented with a circuit of depth $\bigO(k \surd{\tfrac{n}{k}})$ on grid topologies
	of size $\Omega(\surd \tfrac{n}{s}) \times \bigO( \sqrt{ns})$, where $1\leq s \leq k$, 
	with $\bigO(k n)$ 2-qubit gates in total and without ancilla qubits.
	\label{thm:dsu-grid}
\end{theorem}

\begin{proof}[Proof (assuming integer square root values)]
	The first part of our scheme distributes an input Hamming weight $\ell \leq k$ from the top-leftmost rectangle top-to-bottom in the leftmost column. 
	For this we need $\surd\tfrac{n}{k}$ many consecutive weight distribution blocks of the form $\wdb{x}{x-\sqrt{nk}}{k}$, which -- using Lemma~\ref{lem:wdb-LNN} --
	leads to a circuit depth of $\bigO(k \surd \tfrac{n}{k})$ using $\bigO(k^2 \surd \tfrac{n}{k}) \subseteq \bigO(k n)$ many gates and no ancilla qubits.

	The second part of our scheme further distributes weights consecutively from left-to-right in each row using weight distribution blocks of the form $\wdb{y}{y-k}{k}$, parallelizing across rows. 
	This leads to a circuit depth of $\bigO(k \surd \tfrac{n}{k})$ using $\bigO(k^2 \tfrac{n}{k}) = \bigO(k n)$ many gates and no ancilla qubits. 

	The final part again uses Dicke state unitaries $\dsu{k}{k}$ in parallel, not exceeding the previous depth or gate counts, yielding the values in the Theorem. 
\end{proof}

We note that this scheme leaves room for improvement. For one, we could start with the input Hamming weight in a central rectangle, 
from where it would be distributed first along its column and then to the half-rows on the left and on the right, decreasing the overall depth by a factor of 4. 
More importantly, instead of using consecutive applications of $\wdb{y}{y-k}{k}$ along the rows followed by Dicke state unitaries $\dsu{k}{k}$, 
we could directly apply Dicke state unitaries $\dsu{\sqrt{nk}}{k}$ to symmetrically distribute the Hamming weight across all qubits in the $\surd\tfrac{n}{k}$ rectangles of size $k$.
This strategy also allows us to finally deal with non-integer square root values.

%% D8,2 Preparation full connectivity
\renewcommand{\ryangle}[4]{\underset{\scriptstyle\smash{\surd#1/#2}}{#4\theta#3}}
\renewcommand{\rygate}[4]{\gate{\ryangle{#1}{#2}{#3}{#4}}}
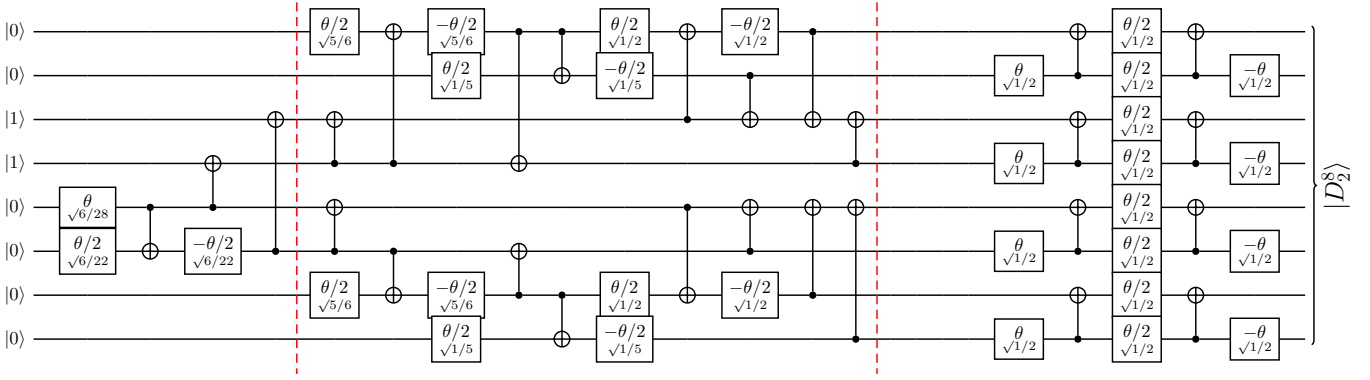
\begin{figure*}[t!]
	\centering
	\begin{adjustbox}{width=\linewidth}
	%% Grover Mixer QAOA
	%% -----------------
	%\hspace*{-0.5cm}
		\begin{quantikz}[row sep={24pt,between origins},execute at end picture={
			}]
			\lstick{\ket{0}}	& \qw			& \qw		& \qw			& \qw\slice{}	& \rygate{5}{6}{/2}{}	& \targ{}	& \rygate{5}{6}{/2}{-}	& \ctrl{3}	& \ctrl{1}	& \rygate{1}{2}{/2}{}	& \targ{}	& \rygate{1}{2}{/2}{-}	& \ctrl{2}	& \qw\slice{}	& \qw	& \qw	& \qw	& \qw	& \qw			& \targ{}	& \rygate{1}{2}{/2}{}	& \targ{}	& \qw			& \qw\rstick[8]{\rotatebox{90}{\Large \dicke{8}{2}}}			\\
			\lstick{\ket{0}}	& \qw			& \qw		& \qw			& \qw		& \qw			& \qw		& \rygate{1}{5}{/2}{}	& \qw		& \targ{}	& \rygate{1}{5}{/2}{-}	& \qw		& \ctrl{1}		& \qw		& \qw		& \qw	& \qw	& \qw	& \qw	& \rygate{1}{2}{}{}	& \ctrl{-1}	& \rygate{1}{2}{/2}{}	& \ctrl{-1}	& \rygate{1}{2}{}{-}	& \qw		\\
			\lstick{\ket{1}}	& \qw			& \qw		& \qw			& \targ{}	& \targ{}		& \qw		& \qw			& \qw		& \qw		& \qw			& \ctrl{-2}	& \targ{}		& \targ{}	& \targ{}	& \qw	& \qw	& \qw	& \qw	& \qw			& \targ{}	& \rygate{1}{2}{/2}{}	& \targ{}	& \qw			& \qw		\\
			\lstick{\ket{1}}	& \qw			& \qw		& \targ{}		& \qw		& \ctrl{-1}		& \ctrl{-3}	& \qw			& \targ{}	& \qw		& \qw			& \qw		& \qw			& \qw		& \ctrl{-1}	& \qw	& \qw	& \qw	& \qw	& \rygate{1}{2}{}{}	& \ctrl{-1}	& \rygate{1}{2}{/2}{}	& \ctrl{-1}	& \rygate{1}{2}{}{-}	& \qw		\\
			\lstick{\ket{0}}	& \rygate{6}{28}{}{}	& \ctrl{1}	& \ctrl{-1}		& \qw		& \targ{}		& \qw		& \qw			& \qw		& \qw		& \qw			& \ctrl{2}	& \targ{}		& \targ{}	& \targ{}	& \qw	& \qw	& \qw	& \qw	& \qw			& \targ{}	& \rygate{1}{2}{/2}{}	& \targ{}	& \qw			& \qw		\\
			\lstick{\ket{0}}	& \rygate{6}{22}{/2}{}	& \targ{}	& \rygate{6}{22}{/2}{-}	& \ctrl{-3}	& \ctrl{-1}		& \ctrl{1}	& \qw			& \targ{}	& \qw		& \qw			& \qw		& \ctrl{-1}		& \qw		& \qw		& \qw	& \qw	& \qw	& \qw	& \rygate{1}{2}{}{}	& \ctrl{-1}	& \rygate{1}{2}{/2}{}	& \ctrl{-1}	& \rygate{1}{2}{}{-}	& \qw		\\
			\lstick{\ket{0}}	& \qw			& \qw		& \qw			& \qw		& \rygate{5}{6}{/2}{}	& \targ{}	& \rygate{5}{6}{/2}{-}	& \ctrl{-1}	& \ctrl{1}	& \rygate{1}{2}{/2}{}	& \targ{}	& \rygate{1}{2}{/2}{-}	& \ctrl{-2}	& \qw		& \qw	& \qw	& \qw	& \qw	& \qw			& \targ{}	& \rygate{1}{2}{/2}{}	& \targ{}	& \qw			& \qw		\\	
			\lstick{\ket{0}}	& \qw			& \qw		& \qw			& \qw		& \qw			& \qw		& \rygate{1}{5}{/2}{}	& \qw		& \targ{}	& \rygate{1}{5}{/2}{-}	& \qw		& \qw			& \qw		& \ctrl{-3}	& \qw	& \qw	& \qw	& \qw	& \rygate{1}{2}{}{}	& \ctrl{-1}	& \rygate{1}{2}{/2}{}	& \ctrl{-1}	& \rygate{1}{2}{}{-}	& \qw							
		\end{quantikz}
	\end{adjustbox}
	\caption{Preparation of a Dicke state \dicke{8}{2} on an all-to-all connectivity: The leftmost weight distribution block $\wdb{8}{4}{2}$ accepts only input Hamming weight $2$.
		It is followed by two recursive weight distribution blocks $\wdb{4}{2}{2}$, and finally by four Dicke state unitaries $\dsu{2}{2}$. The total CNOT count is 27. 
		This improves on previous work with an explicit CNOT gate count of 44~\cite{mukherjee2020preparing}.
	}
	\label{fig:D82-full}
\end{figure*}

\begin{proof}[Full proof of Theorem~\ref{thm:dsu-grid}]
	The main issue with non-integer square root values in the construction behind the previous proof is that a straight-forward adaptation
	can lead to rectangles which contain several unused qubits. While these are not ancilla qubits in the sense that we use them in our circuit design, they can 
	block qubits form interacting in our circuits, and thus need to be swapped out of the way, thus this adds unnecessary ancilla qubits that we have to take care of.

	There is an elegant way to deal with this based on the observation above: We only construct snake-shaped sets of $k$ qubits occupying (parts) of rectangles of size
	$\Omega(\surd\tfrac{k}{s}) \times \bigO(\sqrt{ks})$ \emph{in the leftmost column}. Instead of the other columns of rectangles, we use $\surd\tfrac{n}{k}$ many 
	snake-shaped sets of $\Theta(\sqrt{nk})$ many qubits, which each occupy the space of a row and attach to the heads or tails of the snake in the corresponding rectangle in the leftmost column. 
	
	Then our scheme first distributes an input Hamming weight $\ell \leq k$ from the top rectangle down to other rectangles in the leftmost column,
	with the same gate count as before. 
	This is immediately followed by $\surd \tfrac{n}{k}$ many Dicke state unitaries $\dsu{\bigO(\sqrt{nk})}{k}$ on each pair of (i) a short snake of qubits in a rectangle in the leftmost column and (ii) the long snake of qubits in the row to its right. 
	(In case the snakes are glued together in the wrong order, we first use a SWAP network on the short snake to reverse the order of its qubits.)
	The rows can all be executed in parallel, yielding a circuit depth of $\bigO(\sqrt{nk}) = \bigO(k \surd\tfrac{n}{k})$ using $\bigO(\surd \tfrac{n}{k} k \sqrt{nk}) = \bigO(kn)$ gates in total,
	without the need for (unnecessary) ancilla qubits.
\end{proof}

We remark that for grid topologies of size $\Theta(\surd{\tfrac{n}{k}}) \times \Theta(\sqrt{nk})$, our circuit depth asymptotically corresponds to the graph diameter of the grid topology. We thus conjecture for these grids that our circuit depths are optimal up to constant factors. 
Techniques such as teleportation using mid-circuit measurements and classical feed-forward controls to distant qubits could, however, be used to circumvent these connectivity restrictions, although one might argue that these techniques effectively decrease the graph diameter itself.

%% D8,2 Preparation grid connectivity
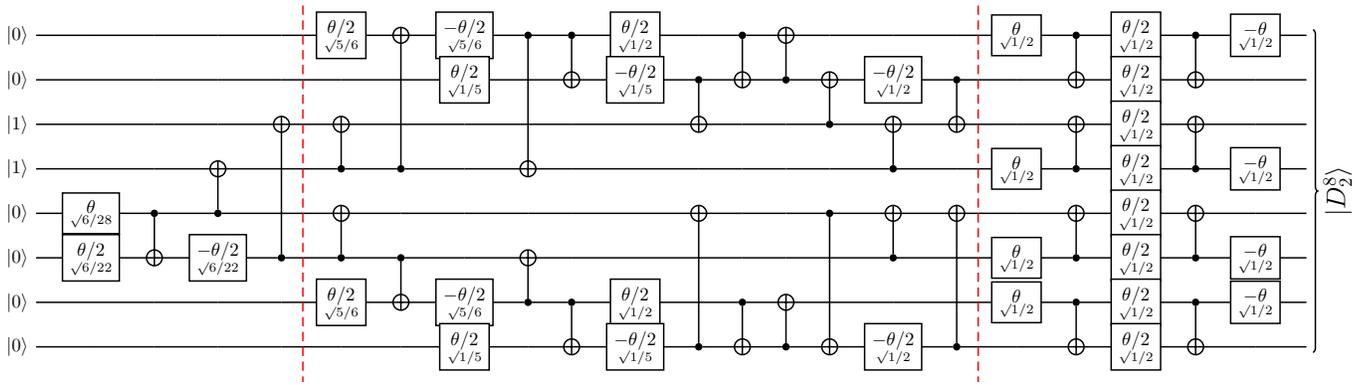
\begin{figure*}[t!]
	\centering
	\begin{adjustbox}{width=\linewidth}
	%% Grover Mixer QAOA
	%% -----------------
	%\hspace*{-0.5cm}
		\begin{quantikz}[row sep={24pt,between origins},execute at end picture={
			}]
			\lstick{\ket{0}}	& \qw			& \qw		& \qw			& \qw\slice{}	& \rygate{5}{6}{/2}{}	& \targ{}	& \rygate{5}{6}{/2}{-}	& \ctrl{3}	& \ctrl{1}	& \rygate{1}{2}{/2}{}	& \qw		& \ctrl{1}	& \targ{}	& \qw		& \qw			& \qw\slice{}		& \rygate{1}{2}{}{}	& \ctrl{1}	& \rygate{1}{2}{/2}{}	& \ctrl{1}	& \rygate{1}{2}{}{-}	& \qw\rstick[8]{\rotatebox{90}{\Large \dicke{8}{2}}}			\\
			\lstick{\ket{0}}	& \qw			& \qw		& \qw			& \qw		& \qw			& \qw		& \rygate{1}{5}{/2}{}	& \qw		& \targ{}	& \rygate{1}{5}{/2}{-}	& \ctrl{1}	& \targ{}	& \ctrl{-1}	& \targ{}	& \rygate{1}{2}{/2}{-}	& \ctrl{1}		& \qw			& \targ{}	& \rygate{1}{2}{/2}{}	& \targ{}	& \qw			& \qw		\\
			\lstick{\ket{1}}	& \qw			& \qw		& \qw			& \targ{}	& \targ{}		& \qw		& \qw			& \qw		& \qw		& \qw			& \targ{}	& \qw		& \qw		& \ctrl{-1}	& \targ{}		& \targ{}		& \qw			& \targ{}	& \rygate{1}{2}{/2}{}	& \targ{}	& \qw			& \qw		\\
			\lstick{\ket{1}}	& \qw			& \qw		& \targ{}		& \qw		& \ctrl{-1}		& \ctrl{-3}	& \qw			& \targ{}	& \qw		& \qw			& \qw		& \qw		& \qw		& \qw		& \ctrl{-1}		& \qw			& \rygate{1}{2}{}{}	& \ctrl{-1}	& \rygate{1}{2}{/2}{}	& \ctrl{-1}	& \rygate{1}{2}{}{-}	& \qw		\\
			\lstick{\ket{0}}	& \rygate{6}{28}{}{}	& \ctrl{1}	& \ctrl{-1}		& \qw		& \targ{}		& \qw		& \qw			& \qw		& \qw		& \qw			& \targ{}	& \qw		& \qw		& \ctrl{3}	& \targ{}		& \targ{}		& \qw			& \targ{}	& \rygate{1}{2}{/2}{}	& \targ{}	& \qw			& \qw		\\
			\lstick{\ket{0}}	& \rygate{6}{22}{/2}{}	& \targ{}	& \rygate{6}{22}{/2}{-}	& \ctrl{-3}	& \ctrl{-1}		& \ctrl{1}	& \qw			& \targ{}	& \qw		& \qw			& \qw		& \qw		& \qw		& \qw		& \ctrl{-1}	      	& \qw			& \rygate{1}{2}{}{}	& \ctrl{-1}	& \rygate{1}{2}{/2}{}	& \ctrl{-1}	& \rygate{1}{2}{}{-}	& \qw		\\
			\lstick{\ket{0}}	& \qw			& \qw		& \qw			& \qw		& \rygate{5}{6}{/2}{}	& \targ{}	& \rygate{5}{6}{/2}{-}	& \ctrl{-1}	& \ctrl{1}	& \rygate{1}{2}{/2}{}	& \qw		& \ctrl{1}	& \targ{}	& \qw		& \qw			& \qw			& \rygate{1}{2}{}{}	& \ctrl{1}	& \rygate{1}{2}{/2}{}	& \ctrl{1}	& \rygate{1}{2}{}{-}	& \qw		\\	
			\lstick{\ket{0}}	& \qw			& \qw		& \qw			& \qw		& \qw			& \qw		& \rygate{1}{5}{/2}{}	& \qw		& \targ{}	& \rygate{1}{5}{/2}{-}	& \ctrl{-3}	& \targ{}	& \ctrl{-1}	& \targ{}	& \rygate{1}{2}{/2}{-}	& \ctrl{-3}		& \qw			& \targ{}	& \rygate{1}{2}{/2}{}	& \targ{}	& \qw			& \qw							
		\end{quantikz}
	\end{adjustbox}
	\caption{Preparation of a Dicke state \dicke{8}{2} on grid connectivity: The leftmost weight distribution block $\wdb{8}{4}{2}$ accepts only input Hamming weight $2$.
		It is followed by two recursive weight distribution blocks $\wdb{4}{2}{2}$, which are combined with a reversal of two wires,  
		and finally by four Dicke state unitaries $\dsu{2}{2}$ (two of them oriented opposite to the others). 
		The total CNOT count is 31. This improves on an extrapolated lower bound of $35$ CNOT gates for grids based on the scheme in~\cite{aktar2022divideconquer}. 
	}
	\label{fig:D82-grid}
\end{figure*}

\section{Conclusion}
\label{sec:conclusion}

We have presented Dicke state preparation circuits of depth $\bigO(k \log \tfrac{n}{k})$ and $\bigO(k \surd\tfrac{n}{k})$ for all-to-all and grid topologies, respectively, which
significantly improve and expand upon previous state-of-the art circuit depths.
It is at least plausible that for some parameter combinations, such as for constant $k$ on all-to-all connectivity or for arbitrary $k$ on $\surd{\tfrac{n}{k}} \times \sqrt{nk}$ grid topologies, there is no further improvement beyond constant factors relying only on quantum gates (i.e. without state or gate teleportation). 

On the other hand, our asymptotic improvements are not the end of the story:
We have observed that our techniques can also yield significantly lower gate counts for specific Dicke states when compared to existing work.
For example, in Figures~\ref{fig:D82-full}~\&~\ref{fig:D82-grid} we present circuits for $\dicke{8}{2}$ preparation on all-to-all and grid connectivity, with a total CNOT gate count of 27 and 31, respectively. 
To the best of our knowledge, both of these values easily lower the record in the number of $2$-qubit gates necessary to prepare $\dicke{8}{2}$.
The circuits can also be opened in Quirk from the following links for
\href{https://algassert.com/quirk#circuit=%7B%22cols%22%3A%5B%5B1%2C1%2C1%2C1%2C%22~1u3%22%2C%22~oppl%22%5D%2C%5B1%2C1%2C1%2C%22X%22%2C%22%E2%80%A2%22%2C%22X%22%5D%2C%5B1%2C1%2C1%2C1%2C1%2C%22~6d8n%22%5D%2C%5B1%2C1%2C%22X%22%2C1%2C1%2C%22%E2%80%A2%22%5D%2C%5B1%2C1%2C%22Chance4%22%5D%2C%5B1%2C1%2C%22X%22%2C%22%E2%80%A2%22%5D%2C%5B%22~f99s%22%5D%2C%5B%22X%22%2C1%2C1%2C%22%E2%80%A2%22%5D%2C%5B%22~eag4%22%2C%22~h68j%22%5D%2C%5B%22%E2%80%A2%22%2C%22X%22%2C1%2C%22X%22%5D%2C%5B%22~l849%22%2C%22~k3b5%22%5D%2C%5B%22X%22%2C1%2C%22%E2%80%A2%22%5D%2C%5B%22~kb4m%22%5D%2C%5B1%2C%22%E2%80%A2%22%2C%22X%22%5D%2C%5B%22%E2%80%A2%22%2C1%2C%22X%22%5D%2C%5B1%2C1%2C%22X%22%2C%22%E2%80%A2%22%5D%2C%5B%22Chance4%22%5D%2C%5B1%2C1%2C1%2C1%2C%22X%22%2C%22%E2%80%A2%22%5D%2C%5B1%2C1%2C1%2C1%2C1%2C1%2C%22~f99s%22%5D%2C%5B1%2C1%2C1%2C1%2C1%2C%22%E2%80%A2%22%2C%22X%22%5D%2C%5B1%2C1%2C1%2C1%2C1%2C1%2C%22~eag4%22%2C%22~h68j%22%5D%2C%5B1%2C1%2C1%2C1%2C1%2C%22X%22%2C%22%E2%80%A2%22%2C%22X%22%5D%2C%5B1%2C1%2C1%2C1%2C1%2C1%2C%22~l849%22%2C%22~k3b5%22%5D%2C%5B1%2C1%2C1%2C1%2C%22%E2%80%A2%22%2C1%2C%22X%22%5D%2C%5B1%2C1%2C1%2C1%2C1%2C1%2C%22~kb4m%22%5D%2C%5B1%2C1%2C1%2C1%2C%22X%22%2C%22%E2%80%A2%22%5D%2C%5B1%2C1%2C1%2C1%2C%22X%22%2C1%2C%22%E2%80%A2%22%5D%2C%5B1%2C1%2C1%2C1%2C%22X%22%2C1%2C1%2C%22%E2%80%A2%22%5D%2C%5B%22Chance4%22%2C1%2C1%2C1%2C%22Chance4%22%5D%2C%5B%22~5npc%22%2C1%2C%22~5npc%22%2C1%2C%22~5npc%22%2C1%2C%22~5npc%22%5D%5D%2C%22gates%22%3A%5B%7B%22id%22%3A%22~suqi%22%2C%22name%22%3A%22%CE%B8%2F2%5Cn%E2%88%9A1%2F2%22%2C%22circuit%22%3A%7B%22cols%22%3A%5B%5B%7B%22id%22%3A%22Ryft%22%2C%22arg%22%3A%22acos(sqrt(1%2F2))%22%7D%5D%5D%7D%7D%2C%7B%22id%22%3A%22~5npc%22%2C%22name%22%3A%22U2%2C2%22%2C%22circuit%22%3A%7B%22cols%22%3A%5B%5B1%2C%7B%22id%22%3A%22Ryft%22%2C%22arg%22%3A%22pi%2F2%22%7D%5D%2C%5B%22X%22%2C%22%E2%80%A2%22%5D%2C%5B%22~suqi%22%2C%22~suqi%22%5D%2C%5B%22X%22%2C%22%E2%80%A2%22%5D%2C%5B1%2C%7B%22id%22%3A%22Ryft%22%2C%22arg%22%3A%22-pi%2F2%22%7D%5D%5D%7D%7D%2C%7B%22id%22%3A%22~1u3%22%2C%22name%22%3A%22%CE%B8%5Cn%E2%88%9A6%2F28%22%2C%22circuit%22%3A%7B%22cols%22%3A%5B%5B%7B%22id%22%3A%22Ryft%22%2C%22arg%22%3A%222acos(sqrt(6%2F28))%22%7D%5D%5D%7D%7D%2C%7B%22id%22%3A%22~oppl%22%2C%22name%22%3A%22%CE%B8%2F2%5Cn%E2%88%9A6%2F22%22%2C%22circuit%22%3A%7B%22cols%22%3A%5B%5B%7B%22id%22%3A%22Ryft%22%2C%22arg%22%3A%22acos(sqrt(6%2F22))%22%7D%5D%5D%7D%7D%2C%7B%22id%22%3A%22~6d8n%22%2C%22name%22%3A%22-%CE%B8%2F2%5Cn%E2%88%9A6%2F22%22%2C%22circuit%22%3A%7B%22cols%22%3A%5B%5B%7B%22id%22%3A%22Ryft%22%2C%22arg%22%3A%22-acos(sqrt(6%2F22))%22%7D%5D%5D%7D%7D%2C%7B%22id%22%3A%22~f99s%22%2C%22name%22%3A%22%CE%B8%2F2%5Cn%E2%88%9A5%2F6%22%2C%22circuit%22%3A%7B%22cols%22%3A%5B%5B%7B%22id%22%3A%22Ryft%22%2C%22arg%22%3A%22acos(sqrt(5%2F6))%22%7D%5D%5D%7D%7D%2C%7B%22id%22%3A%22~eag4%22%2C%22name%22%3A%22-%CE%B8%2F2%5Cn%E2%88%9A5%2F6%22%2C%22circuit%22%3A%7B%22cols%22%3A%5B%5B%7B%22id%22%3A%22Ryft%22%2C%22arg%22%3A%22-acos(sqrt(5%2F6))%22%7D%5D%5D%7D%7D%2C%7B%22id%22%3A%22~h68j%22%2C%22name%22%3A%22%CE%B8%2F2%5Cn%E2%88%9A1%2F5%22%2C%22circuit%22%3A%7B%22cols%22%3A%5B%5B%7B%22id%22%3A%22Ryft%22%2C%22arg%22%3A%22acos(sqrt(1%2F5))%22%7D%5D%5D%7D%7D%2C%7B%22id%22%3A%22~k3b5%22%2C%22name%22%3A%22-%CE%B8%2F2%5Cn%E2%88%9A1%2F5%22%2C%22circuit%22%3A%7B%22cols%22%3A%5B%5B%7B%22id%22%3A%22Ryft%22%2C%22arg%22%3A%22-acos(sqrt(1%2F5))%22%7D%5D%5D%7D%7D%2C%7B%22id%22%3A%22~l849%22%2C%22name%22%3A%22%CE%B8%2F2%5Cn%E2%88%9A1%2F2%22%2C%22circuit%22%3A%7B%22cols%22%3A%5B%5B%7B%22id%22%3A%22Ryft%22%2C%22arg%22%3A%22acos(sqrt(1%2F2))%22%7D%5D%5D%7D%7D%2C%7B%22id%22%3A%22~kb4m%22%2C%22name%22%3A%22-%CE%B8%2F2%5Cn%E2%88%9A1%2F2%22%2C%22circuit%22%3A%7B%22cols%22%3A%5B%5B%7B%22id%22%3A%22Ryft%22%2C%22arg%22%3A%22-acos(sqrt(1%2F2))%22%7D%5D%5D%7D%7D%5D%2C%22init%22%3A%5B0%2C0%2C1%2C1%5D%7D
}{all-to-all connectivity} %
and for  
\href{https://algassert.com/quirk#circuit=%7B%22cols%22%3A%5B%5B1%2C1%2C1%2C1%2C%22~1u3%22%2C%22~oppl%22%5D%2C%5B1%2C1%2C1%2C%22X%22%2C%22%E2%80%A2%22%2C%22X%22%5D%2C%5B1%2C1%2C1%2C1%2C1%2C%22~6d8n%22%5D%2C%5B1%2C1%2C%22X%22%2C1%2C1%2C%22%E2%80%A2%22%5D%2C%5B1%2C1%2C%22Chance4%22%5D%2C%5B1%2C1%2C%22X%22%2C%22%E2%80%A2%22%5D%2C%5B%22~f99s%22%5D%2C%5B%22X%22%2C1%2C1%2C%22%E2%80%A2%22%5D%2C%5B%22~eag4%22%2C%22~h68j%22%5D%2C%5B%22%E2%80%A2%22%2C%22X%22%2C1%2C%22X%22%5D%2C%5B%22~l849%22%2C%22~k3b5%22%5D%2C%5B1%2C%22%E2%80%A2%22%2C%22X%22%5D%2C%5B%22%E2%80%A2%22%2C%22X%22%5D%2C%5B%22X%22%2C%22%E2%80%A2%22%5D%2C%5B1%2C%22X%22%2C%22%E2%80%A2%22%5D%2C%5B1%2C%22~kb4m%22%5D%2C%5B1%2C1%2C%22X%22%2C%22%E2%80%A2%22%5D%2C%5B1%2C%22%E2%80%A2%22%2C%22X%22%5D%2C%5B%22Chance4%22%5D%2C%5B1%2C1%2C1%2C1%2C%22X%22%2C%22%E2%80%A2%22%5D%2C%5B1%2C1%2C1%2C1%2C1%2C1%2C%22~f99s%22%5D%2C%5B1%2C1%2C1%2C1%2C1%2C%22%E2%80%A2%22%2C%22X%22%5D%2C%5B1%2C1%2C1%2C1%2C1%2C1%2C%22~eag4%22%2C%22~h68j%22%5D%2C%5B1%2C1%2C1%2C1%2C1%2C%22X%22%2C%22%E2%80%A2%22%2C%22X%22%5D%2C%5B1%2C1%2C1%2C1%2C1%2C1%2C%22~l849%22%2C%22~k3b5%22%5D%2C%5B1%2C1%2C1%2C1%2C%22X%22%2C1%2C1%2C%22%E2%80%A2%22%5D%2C%5B1%2C1%2C1%2C1%2C1%2C1%2C%22%E2%80%A2%22%2C%22X%22%5D%2C%5B1%2C1%2C1%2C1%2C1%2C1%2C%22X%22%2C%22%E2%80%A2%22%5D%2C%5B1%2C1%2C1%2C1%2C%22%E2%80%A2%22%2C1%2C1%2C%22X%22%5D%2C%5B1%2C1%2C1%2C1%2C1%2C1%2C1%2C%22~kb4m%22%5D%2C%5B1%2C1%2C1%2C1%2C%22X%22%2C%22%E2%80%A2%22%5D%2C%5B1%2C1%2C1%2C1%2C%22X%22%2C1%2C1%2C%22%E2%80%A2%22%5D%2C%5B%22Chance4%22%2C1%2C1%2C1%2C%22Chance4%22%5D%2C%5B%22Swap%22%2C%22Swap%22%5D%2C%5B1%2C1%2C1%2C1%2C1%2C1%2C%22Swap%22%2C%22Swap%22%5D%2C%5B%22~5npc%22%2C1%2C%22~5npc%22%2C1%2C%22~5npc%22%2C1%2C%22~5npc%22%5D%5D%2C%22gates%22%3A%5B%7B%22id%22%3A%22~suqi%22%2C%22name%22%3A%22%CE%B8%2F2%5Cn%E2%88%9A1%2F2%22%2C%22circuit%22%3A%7B%22cols%22%3A%5B%5B%7B%22id%22%3A%22Ryft%22%2C%22arg%22%3A%22acos(sqrt(1%2F2))%22%7D%5D%5D%7D%7D%2C%7B%22id%22%3A%22~5npc%22%2C%22name%22%3A%22U2%2C2%22%2C%22circuit%22%3A%7B%22cols%22%3A%5B%5B1%2C%7B%22id%22%3A%22Ryft%22%2C%22arg%22%3A%22pi%2F2%22%7D%5D%2C%5B%22X%22%2C%22%E2%80%A2%22%5D%2C%5B%22~suqi%22%2C%22~suqi%22%5D%2C%5B%22X%22%2C%22%E2%80%A2%22%5D%2C%5B1%2C%7B%22id%22%3A%22Ryft%22%2C%22arg%22%3A%22-pi%2F2%22%7D%5D%5D%7D%7D%2C%7B%22id%22%3A%22~1u3%22%2C%22name%22%3A%22%CE%B8%5Cn%E2%88%9A6%2F28%22%2C%22circuit%22%3A%7B%22cols%22%3A%5B%5B%7B%22id%22%3A%22Ryft%22%2C%22arg%22%3A%222acos(sqrt(6%2F28))%22%7D%5D%5D%7D%7D%2C%7B%22id%22%3A%22~oppl%22%2C%22name%22%3A%22%CE%B8%2F2%5Cn%E2%88%9A6%2F22%22%2C%22circuit%22%3A%7B%22cols%22%3A%5B%5B%7B%22id%22%3A%22Ryft%22%2C%22arg%22%3A%22acos(sqrt(6%2F22))%22%7D%5D%5D%7D%7D%2C%7B%22id%22%3A%22~6d8n%22%2C%22name%22%3A%22-%CE%B8%2F2%5Cn%E2%88%9A6%2F22%22%2C%22circuit%22%3A%7B%22cols%22%3A%5B%5B%7B%22id%22%3A%22Ryft%22%2C%22arg%22%3A%22-acos(sqrt(6%2F22))%22%7D%5D%5D%7D%7D%2C%7B%22id%22%3A%22~f99s%22%2C%22name%22%3A%22%CE%B8%2F2%5Cn%E2%88%9A5%2F6%22%2C%22circuit%22%3A%7B%22cols%22%3A%5B%5B%7B%22id%22%3A%22Ryft%22%2C%22arg%22%3A%22acos(sqrt(5%2F6))%22%7D%5D%5D%7D%7D%2C%7B%22id%22%3A%22~eag4%22%2C%22name%22%3A%22-%CE%B8%2F2%5Cn%E2%88%9A5%2F6%22%2C%22circuit%22%3A%7B%22cols%22%3A%5B%5B%7B%22id%22%3A%22Ryft%22%2C%22arg%22%3A%22-acos(sqrt(5%2F6))%22%7D%5D%5D%7D%7D%2C%7B%22id%22%3A%22~h68j%22%2C%22name%22%3A%22%CE%B8%2F2%5Cn%E2%88%9A1%2F5%22%2C%22circuit%22%3A%7B%22cols%22%3A%5B%5B%7B%22id%22%3A%22Ryft%22%2C%22arg%22%3A%22acos(sqrt(1%2F5))%22%7D%5D%5D%7D%7D%2C%7B%22id%22%3A%22~k3b5%22%2C%22name%22%3A%22-%CE%B8%2F2%5Cn%E2%88%9A1%2F5%22%2C%22circuit%22%3A%7B%22cols%22%3A%5B%5B%7B%22id%22%3A%22Ryft%22%2C%22arg%22%3A%22-acos(sqrt(1%2F5))%22%7D%5D%5D%7D%7D%2C%7B%22id%22%3A%22~l849%22%2C%22name%22%3A%22%CE%B8%2F2%5Cn%E2%88%9A1%2F2%22%2C%22circuit%22%3A%7B%22cols%22%3A%5B%5B%7B%22id%22%3A%22Ryft%22%2C%22arg%22%3A%22acos(sqrt(1%2F2))%22%7D%5D%5D%7D%7D%2C%7B%22id%22%3A%22~kb4m%22%2C%22name%22%3A%22-%CE%B8%2F2%5Cn%E2%88%9A1%2F2%22%2C%22circuit%22%3A%7B%22cols%22%3A%5B%5B%7B%22id%22%3A%22Ryft%22%2C%22arg%22%3A%22-acos(sqrt(1%2F2))%22%7D%5D%5D%7D%7D%5D%2C%22init%22%3A%5B0%2C0%2C1%2C1%5D%7D
}{grid connectivity}.
However, these circuits rely on hand-tuned non-scalable implementations of our weight distribution blocks. 
This leaves room for further improvements in the sense of scalable yet resource-efficient implementations of the the latter.

We repeat the observations from Sections~\ref{sec:intro}~and~\ref{sec:preliminaries} that 
(i) our results for Dicke states $\dicke{n}{k}$ also hold for Dicke states $\dicke{n}{n-k} = X^{\otimes n}\dicke{n}{k}$, 
(ii) due to the variable input structure of Dicke state unitaries $\dsu{n}{k}$
our results extend to Symmetric states consisting only of computational basis states with Hamming weight $\ell \leq k$. 
Combining the two, we also get circuits preparing Symmetric states consisting only of computational basis states with Hamming weight $\ell \geq n-k$.

Finally, using our circuits in reverse with some circuit additions responsible for implementing encoding changes~\cite{baertschi2019deterministic}, 
we get $\bigO(k \log \tfrac{n}{k})$-depth circuits for quantum compression 
of the discussed types of Hamming-weight restricted Symmetric states from $n$ into $\lceil \log(k+1) \rceil$ many qubits.

%% LEAVE PLAINURL FOR REVIEW PHASE
\nocite{quantikz}
\bibliographystyle{plainurl}%IEEEtran} % prefer plainurl over IEEEtran, as it shows DOIs and arXiv links. 
\bibliography{ds-bib}

\end{document}